%% file: main.tex
\documentclass[a4paper,UKenglish,cleveref, autoref, thm-restate]{lipics-v2021}

\usepackage{amsmath}
\usepackage{amsthm}
\usepackage{graphicx} 
\usepackage{xcolor}
\usepackage{amsfonts}
\usepackage{bm}
\usepackage{caption}
\usepackage{pifont}
\usepackage{dsfont}
\usepackage[linesnumbered]{algorithm2e}
\RestyleAlgo{ruled}
\usepackage{subcaption}
\usepackage[commandnameprefix=always]{changes}
\usepackage{todonotes,soul}
\usepackage{wrapfig}

\input{macros}

\title{Highly Dynamic and Fully Distributed Data Structures} 

 \author{John Augustine}{Indian Institute of Technology Madras, Chennai, India }{augustine@cse.iitm.ac.in}{https://orcid.org/0000-0003-0948-3961}{}
\author{Antonio Cruciani}{Aalto University, Espoo, Finland }{antonio.cruciani@aalto.fi}{https://orcid.org/0000-0002-9538-4275}{This work was supported in part by the Research Council of Finland, Grant 363558 and Partially supported by the Cryptography, Cybersecutiry, and Distributed
	Trust laboratory at IIT Madras (Indian Institute of Technology Madras) while visiting the institute.}

\author{Iqra Altaf Gillani}{National Institute of Technology Srinagar, Srinagar, India }{iqraaltaf@nitsri.ac.in}{https://orcid.org/0000-0001-8656-4023}{}

\authorrunning{Augustine et al.} 

\Copyright{John Augustine and Antonio Cruciani and Iqra Altaf Gillani} 



\category{} 

\relatedversion{} 

\acknowledgements{}

\nolinenumbers 

\EventEditors{}
\EventNoEds{2}
\EventLongTitle{}
\EventShortTitle{}
\EventAcronym{}
\EventYear{}
\EventDate{}
\EventLocation{}
\EventLogo{}
\SeriesVolume{42}
\ArticleNo{23}

\begin{document}

\maketitle

\input{trunk/abstract}
\newpage
\tableofcontents
\newpage
\input{trunk/intro_antonio}
\input{trunk/architecture}

\input{trunk/generalization}
\input{trunk/general_keys}
\input{trunk/implications}

\input{trunk/related_works}
\input{trunk/conclusion}



\bibliography{references}

\appendix
\begin{center}
	\LARGE{\textbf{Appendix}}
\end{center} 
\input{trunk/appendix}

\end{document}

%% file: macros.tex
\def\nccz {$\text{NCC}_0$}
\def\polylog {\text{polylog}}

\def\whp {{\text{w.h.p.}}}

\usepackage[normalem]{ulem}

\newcommand{\Poly}[1]{poly(#1)}

\newcommand{\expect}[1]{\mathbb{E}\left[#1\right]}
\newcommand{\prob}[1]{\textbf{\textsc{Pr}}\left(#1\right)}

\newcommand{\Spartan}{\text{Spartan}}
\newcommand{\Ll}{\mathcal{L}}
\newcommand{\LL}[1]{\mathcal{L}^{(#1)}}
\newcommand{\A}{\mathcal{A}}
\newcommand{\B}{\mathcal{B}}
\newcommand{\BB}[1]{\mathcal{B}^{(#1)}}

\newcommand{\C}{\mathcal{C}}
\newcommand{\CC}[1]{\mathcal{C}^{(#1)}}

\newcommand{\Oo}{\mathcal{O}}

\newcommand{\M}{\mathcal{M}}
\newcommand{\W}{\mathcal{W}}
\newcommand{\V}{\mathcal{V}}

\newcommand{\X}{\mathcal{S}}
\newcommand{\XX}[1]{\mathcal{S}^{(#1)}}

\newcommand{\skippl}{\texttt{SKIP}$^+$}
\newcommand{\neig}[4]{N_{#1}^{(#2,#3)}(#4)}

\newcommand{\Polylog}[1]{\text{polylog}{(#1)}}
\newcommand{\lmax}[2]{\ell_{\max}^{#1}(#2)}
\newcommand{\Lmax}[1]{\ell_{\max}(#1)}
\newcommand{\childrenLevel}[3]{\Gamma_{#1}^{#2}(#3)}

\newcommand{\ChildrenLevel}[2]{\Gamma_{#1}{(#2)}}
\newcommand{\Children}[1]{\Gamma{(#1)}}

\newcommand{\nosemic}{\renewcommand{\@endalgocfline}{\relax}}
\newcommand{\dosemic}{\renewcommand{\@endalgocfline}{\algocf@endline}}
\let\oldnl\nl
\newcommand{\nonl}{\renewcommand{\nl}{\let\nl\oldnl}}

\def\ShowComment{True}
\ifdefined\ShowComment

\else

\fi

\newboolean{short}
\setboolean{short}{false}

\newcommand{\shortOnly}[1]{\ifthenelse{\boolean{short}}{#1}{}}
\newcommand{\onlyShort}[1]{\ifthenelse{\boolean{short}}{}{#1}}
\newcommand{\longOnly}[1]{\ifthenelse{\boolean{short}}{}{#1}}
\newcommand{\onlyLong}[1]{\ifthenelse{\boolean{short}}{}{#1}}

%% file: trunk/abstract.tex
\begin{abstract}

We study robust and efficient distributed algorithms for building and maintaining distributed data structures in dynamic Peer-to-Peer (P2P) networks. P2P networks are characterized by a high level of dynamicity with abrupt heavy node \emph{churn} (nodes that join and leave the network continuously over time). We present a novel algorithmic framework to build and maintain, with high probability, a skip list for $\text{poly}(n)$ rounds despite a \emph{churn rate} of $O(n/\log n)$, which is the number of nodes joining and/or leaving \emph{per round}; $n$ is the stable network size. We assume that the churn is controlled by an oblivious adversary that has complete knowledge and control of what nodes join and leave and at what time and has unlimited computational power, but is oblivious to the random choices made by the algorithm. Importantly, the maintenance overhead in any interval of time (measured in terms of the total number of messages exchanged and the number of edges formed/deleted) is (up to log factors) proportional to the churn rate. Furthermore, the algorithm is scalable in that the messages are small (i.e., at most $\text{polylog}(n)$ bits) and every node sends and receives at most $\text{polylog}(n)$ messages per round.

To the best of our knowledge, our work provides the first-known fully-distributed data structure and associated algorithms that provably work under highly dynamic settings (i.e., high churn rate that is near-linear in $n$). Furthermore, the nodes operate in a localized manner.

    Our framework crucially relies on new distributed and parallel algorithms to merge two $n$-element skip lists and delete a large subset of items, both in $\Oo(\log n)$ rounds with high probability. These procedures may be of independent interest due to their elegance and potential applicability in other contexts in distributed data structures.

Finally, we believe that our framework can be generalized to other distributed and dynamic data structures including graphs, potentially leading to stable distributed computation despite heavy churn.

\end{abstract}

%% file: trunk/intro_antonio.tex
\section{Introduction}
Peer-to-peer (P2P) computing is emerging as one of the key networking technologies in recent years with many application systems. These have been used to provide distributed resource sharing, storage, messaging, and content streaming, e.g., Gnutella~\cite{Gnutella}, Skype~\cite{Skype}, BitTorrent~\cite{Bittorrent}, ClashPlan~\cite{Crashplan}, Symform~\cite{Sysform}, and Signal~\cite{Signal}. P2P networks are intrinsically highly dynamic networks characterized by a high degree of node \emph{churn} i.e., nodes continuously joining and leaving the network. Connections (edges) may be added or deleted at any time and thus the topology abruptly changes. Moreover, empirical measurements of real-world P2P networks~\cite{Falkner_2007,Gummadi_2002,Sen_2004,Stutzbach_2006} show that the churn rate is very high: nearly $50\%$ of the peers in real-world networks are replaced within an hour. Interestingly, despite a large churn rate, these measurements show that the size of the network remains relatively \emph{stable}. 

P2P networks and algorithms have been proposed for a wide variety of tasks such as data storage and retrieval~\cite{Peddemors_2010,Rowstron_2001,Druschel_2001,Hasan_2005}, collaborative filtering~\cite{Canny_2002}, spam detection~\cite{Cloudmark}, data mining~\cite{Datta_2006}, worm detection and suppression~\cite{Vlachos_2004,Malan_2005}, privacy protection of archived data~\cite{Geambasu_2009}, and for cloud computing services~\cite{Sysform,Babaoglu_2012}. Several works proposed  efficient implementations of distributed data structures with low maintenance time (searching, inserting, and deleting elements) and congestion. These include different versions of distributed hash tables (DHT) like CAN~\cite{Ratnasamy_2001}, Chord~\cite{Stoica_2003}, Pastry~\cite{Rowstron_2001}, and Tapestry~\cite{Zhao_2002}. Such distributed data structures have good load-balancing properties but offer no control over where the data is stored. Also, these show partial resilience to node failures. 

To deal with more structured data in P2P networks several distributed data structures have been developed such as Skip Graphs (Aspnes and Shah~\cite{Aspnes_2007}), SkipNets (Harvey et al.,~\cite{Harvey_2003}), Rainbow Skip graphs (Goodrich et al.,~\cite{Goodrich_2006}), and Skip$^+$ (Jacob et al.,~\cite{Jacob_2014}). They have been formally shown to be resilient to a limited number of faults (or equivalently small amounts of churn).
However, none of these data structures have theoretical guarantees of being able to work in a dynamic network with a very high adversarial churn rate, which can be as much as near-linear (in the network size) per round. This can be seen as a major bottleneck in the implementation and use of data structures for P2P systems. Furthermore, several works deal with the problem of the maintenance of a specific graph topology~\cite{Augustine_2015_b,Pandurangan_2016,Drees_2016,Augustine_2021}, solve the agreement problem~\cite{Augustine_2013_a}, elect a leader~\cite{Augustine_2015_a}, and storage and search of data~\cite{Augustine_2013_b} under adversarial churn. Unfortunately, these structures are not conducive for efficient searching and querying.

In this paper, we take a step towards designing provably robust and scalable distributed data structures and concomitant algorithms for large-scale dynamic P2P networks.  More precisely, we focus on the fundamental problem of maintaining a distributed skip list data structure in P2P networks. Many distributed implementations of data structures inspired by skip lists have been proposed to deal with nodes leaving and joining the network. Unfortunately, a common major drawback among all these approaches is the lack of provable resilience against heavy churn. The problem is especially challenging since the goal is to guarantee that, under a high churn rate, the data structure must (i) be able to preserve its overall structure (ii) quickly update the structure after insertions/deletions, and (iii) correctly answer queries. In such a highly dynamic setting, it is non-trivial to even guarantee that a query can ``go through'' the skip list, the churn can simply remove a large fraction of nodes in just one time-step and stop or block the query. On the other hand, it is prohibitively expensive to rebuild the data structure from scratch whenever large number of nodes leave and a new set of nodes join. Thus we are faced with the additional challenge of ensuring that the maintenance overhead is proportional to the number of nodes that leave/join. In a nutshell, our goal is to design and implement distributed data structures that are resilient to heavy adversarial churn without compromising simplicity or scalability.

\subsection{Model: Dynamic Networks with Churn}
Before we formally state our main result, we discuss our dynamic network with churn (DNC) model, which is used in previous works to model peer-to-peer networks in which nodes can be added and deleted at each round by an adversary (see e.g.~\cite{Augustine_2012,Augustine_2015_b}).


We consider a synchronous dynamic network controlled by an \emph{oblivious} adversary, i.e., the adversary does not know the random choices made by the nodes.
The adversary fixes a dynamically changing sequence of sets of network nodes $\V =\left(V_0,V_1,V_2,\dots, V_t,\dots\right)$ where $V_t\subset U$, for some universe of nodes $U$ and $t\geq 1$, denotes the set of nodes present in the network during round $t$. A node $u$ such that $u\in V_t$ and $u\notin V_{t+1}$ is said to be \emph{leaving} at time $t+1$. Similarly, a node $v\notin V_t$ and $v\in V_{t+1}$ is said to be \emph{joining} the network at time $t+1$. Each node has a unique ID and we simply use the same notation (say, $u$) to denote both the node as well as its ID. The lifetime of a node $u$ is (adversarially chosen to be) a pair $(s_u, t_u)$, where $s_u$ refers to its {\em start time} and $t_u$ refers to its \emph{termination time}. The size of the vertex set is assumed to be stable $|V_t| = n$ for all $t$; this assumption can be relaxed to consider a network that can shrink and grow arbitrarily as discussed in Section~\ref{sec:churn_resilient_network}. Each node in $U$ is assumed to have a unique ID chosen from an $ID$ space of size polynomial in $n$. Moreover, for the first $B = \beta \log n$ rounds (for a sufficiently large constant $\beta>0$), called the \emph{bootstrap phase}, the adversary is \emph{silent}, i.e., there is no churn; more precisely, $V_0 =V_1 =\dots =V_B$. We can think of the bootstrap phase as an initial period of stability during which the protocol prepares itself for a harsher form of dynamism. Subsequently, the network is said to be in \emph{maintenance phase} during which $\V$ can experience churn in the sense that a large number of nodes might join and leave dynamically at each time step.

Communication is via message passing. Nodes can send messages of size $\Oo(\Polylog{n})$ bits to each other if they know their IDs, but no more than $\Oo(\Polylog{n})$ incoming and outgoing messages are allowed at each node per round.
Furthermore, nodes can create and delete edges over which messages can be sent/received. This facilitates the creation of structured communication networks. A bidirectional edge $e=(u,v)$ is formed when one of the end points $u$ sends $v$ an invitation message followed by an acceptance message from $v$. The edge $e$ can be deleted when either $u$ or $v$ sends a delete message. Of course, $e$ will be deleted if either $u$ or $v$ leaves the network.

During the maintenance phase, the adversary can apply a churn up to $\Oo(n/\log n)$ nodes per round.
More precisely, for all $t\geq B$, $|V_{t}\setminus V_{t+1}| = |V_{t+1}\setminus V_t| \in \Oo(n/\log n)$ and 
the adversary is only required to ensure that any new node that joins the network must be connected to a distinct pre-existing node in the network; this is to avoid too many nodes being attached to the same node, thereby causing congestion issues. 

Each node can store data items. We wish to maintain them in the form of a suitable dynamic and distributed data structure. For simplicity in exposition, we assume that each node $u$ has one data item which is also its ID. This coupling of the node, its ID, and its data allows us to refer to them interchangeably as either node $u$ or data item $u$. In Section~\ref{sec:generalize_keys}, we will discuss how this coupling assumption can be relaxed to allow a node to contain multiple disparate data items.

\subsection{Problem Statement}
Our primary goal  is to build and maintain an \emph{approximate distributed data structure} of all the items/nodes in the network. The data structure takes at most $I$ rounds to insert a data item $u$ and at most $R$ rounds to remove it. $I$ and $R$ are parameters that we would like to minimize. The distributed data structure should be able to answer membership queries.
Specifically, each query $q(x, r, s)$ is initiated at some source node $s$ in round $r$ and asks whether data item $x$ is present in the network currently. This implies that node $s$ now wishes to know if some node in the network has the value $x$. The query should be answered by round $r+Q$ where $Q$ is the \emph{query time}, i.e., the time to respond to queries. If $t_s > r+Q$, i.e., node $s$ is in the network until round $r+Q$, it must receive the response. We are required to give the guarantee that the query will be answered correctly as long as either (i) there is a node with associated value $x$ whose effective lifetime subsumes the time range $[r, r+Q]$ (in which case the query must be answered affirmatively) or (ii) there is no node with value $x$ whose effective lifetime has any overlap with $[r, r+Q]$ (in which case the query must be answered negatively). In all other cases, we allow queries to be answered arbitrarily. 

Unlike classical data structures, which guarantee strong consistency, our construction deliberately prioritizes scalability and liveness under high churn. As a result, it provides \emph{eventual consistency}, meaning that queries will reflect the correct state once the system has stabilized, but temporary inconsistencies may arise during periods of high churn. We refer to this design as an \emph{approximate distributed data structure}, reflecting its structural inspiration and relaxed consistency guarantees. This tradeoff is necessary to maintain efficiency and fault-tolerance in the presence of high adversarial churn.


In addition, all our algorithms must satisfy a dynamic notion of \emph{resource-competitiveness}~\cite{Bender_2015}. Let $T$ be any interval between two time instants $t_s$ and $t_e$. We require the work in the interval $T$ to be proportional (within $\Polylog{n}$ factors) to the amount of churn experienced from time $t_s-\Oo(\log n)$ to $t_e$. Formally, define the amount of work $W_t$ as the overall number of exchanged messages plus newly formed edges among nodes at time $t$. Let $\W(t_s,t_e) = \sum_{i = {t_s}}^{t_e} W_i$ and $C(t_s,t_e)$ to be the \emph{amount of work} and churn experienced in the interval of time between $t_s$ and $t_e$, respectively. We require that $\W(t_s,t_e) \in \tilde{\Oo}\left(C(t_s-\Oo(\log n),t_e)\right)$ \whp\ for any interval of time $t_s$, $t_e$, where with the notation $\tilde{\Oo}(\cdot)$ we ignore $\Polylog{n}$ factors. Formally,
\begin{figure}[htb!]
	\centering
	\includegraphics[scale = 0.6]{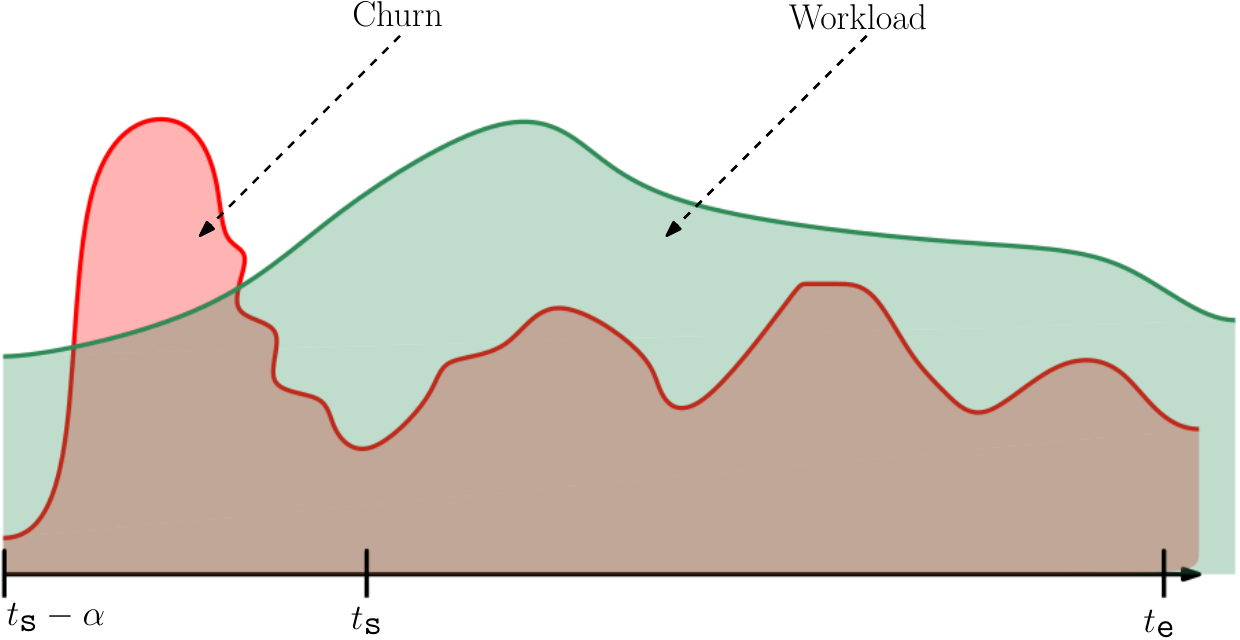}
\caption{Visual representation of our definition of dynamic work efficiency.}\label{fig:dynamic_work_efficiency}
\end{figure}
\begin{definition}[Dynamic Resource Competitiveness]\label{def:dynamic_resource_comp}
	An algorithm $\mathcal{A}$ is $(\alpha,\beta)$-dynamic resource competitive if for any time instants $t_s$ and $t_e$ such that $t_e> t_s$, $\W(t_s,t_e)\in\Oo(\beta C(t_s-\alpha,t_e))$.
\end{definition}
Figure~\ref{fig:dynamic_work_efficiency} is a visual representation of the $(\alpha,\beta)$-dynamic resource competitiveness. The red area represents the amount of churn experienced over time while the green area represent the workload of the algorithm to cope with the churn experienced so far. Figure~\ref{fig:dynamic_work_efficiency} illustrates a spike in churn rate just prior to $t_s$ that resulted in an increased workload after $t_s$. Thus, to capture such delayed effects, the dynamic resource competitiveness definition considers  $C(t_s-\alpha,t_e))$ instead of just $C(t_s,t_e))$. In this work, whenever we refer to an algorithm as dynamic resource competitive, we mean that the algorithm is $(\alpha,\beta)$-dynamic resource competitive with $\alpha =\Oo(\log n)$ and $\beta = \Polylog{n}$. 

\subsection{Our Contributions}\label{sec:contributions}
\subsubsection[Main Result]{\textbf{Main Result}}
We address the problem by presenting a rigorous theoretical framework for the
construction and maintenance of distributed skip lists in highly dynamic distributed systems that can experience heavy churn. 

Any query to a data structure will require some response time, which we naturally wish to minimize. Furthermore, queries will be  imprecise to some extent if the data structure is dynamic. To see why, consider a membership query of the form ``is node $u$ in the dynamic skip list?''  initiated in some round $r$ asking if a node $u$ is in the network or not. Consider a situation where $u$ 
 is far away from the node at which the query was initiated. Suppose $u$ is then deleted shortly thereafter. Although $u$ was present at round $r$, it may or may not be gone when the query procedure reached $u$. Such ambivalences are inevitable, but we wish to limit them. Thus, we define an efficiency parameter $Q$ such that (i) queries raised at round $r$ are answered by round $r+Q$ and (ii) response must be correct in the sense that it must be ``Yes'' (resp., ``No'')  if $u$ was present (resp., not present) from round $r$ to $r+Q$. If $u$ was only present for a portion of the time between $r$ and $r+Q$, then either of the two answers is acceptable. Quite naturally, we wish to minimize $Q$ and for our dynamic data structure, we show that $Q \in \Oo(\log n)$.

Our algorithms 
ensure that the resilient skip list is maintained effectively for at least $\Poly{n}$ rounds with high probability (i.e., with probability $1-1/n^{\Omega(1)}$) even under high adversarial churn. 
Moreover, the overall communication and computation cost incurred by our algorithms is proportional (up to $\Polylog{n}$ factors) to the churn rate, and every node sends and receives at most $\Oo(\Polylog{n})$ messages per round.
In particular, we present the following results (the in-depth descriptions are given in Section~\ref{sec:architecture}):
\begin{enumerate}
	\item An algorithm that constructs and maintains a skip list in a dynamic P2P network with an adversarial churn rate up to $\Oo(n/\log n)$ \emph{per round}.
	\item A novel distributed and parallel algorithm to merge a skip list $\B$ with a base skip list $\C$ in logarithmic time, logarithmic number of messages at every round and an overall amount of work proportional to the size of $\B$, i.e., to the skip list that must be merged with the base one.  While this merge procedure serves as a crucial subroutine in our maintenance procedure, we believe that it is also of independent interest and could potentially find application in other contexts as well. For example, it could be used to speed up the insertion of a batch of elements in skip list-like data structures. Similarly, we designed an efficient distributed and parallel algorithm to delete a batch of elements from a skip list in logarithmic time with overhead proportional to the size of the batch. 
	\item A general framework that we illustrate using skip lists, but can serve as a building block for other complex distributed data structures in highly dynamic networks (see Section~\ref{sec:generalize_ds} for more details).
\end{enumerate}
To the best of our knowledge, our approach is the first fully-distributed skip list data structure that works under highly dynamic settings (high churn rates per step). Furthermore, all the proposed algorithms are localized, easy to implement and scalable. Our major contribution can be summarized in the following theorem.
\begin{theorem}[Main Theorem]\label{thm:overall_result}
	Given a dynamic set of peers initially connected in some suitable manner (e.g., as a single path) that is stable for an initial period of $\Oo(\log n)$ rounds (i.e., the so-called {\em bootstrap phase}) and subsequently experiencing heavy adversarial churn at a churn rate of up to $\Oo(n/\log n)$ nodes joining/leaving per round, there exists a resilient skip list, a distributed data structure that can withstand heavy adversarial churn at a rate of up to $\Oo(n/\log n)$ nodes joining/leaving per round. We provide the following algorithms to support this data structure.
 \begin{itemize}
     \item 
     An $\Oo(\log n)$ round algorithm to construct the resilient skip list during the bootstrap phase,

 \item a fully distributed algorithm that maintains the resilient skip list as nodes are inserted into or deleted from the data structure with the guarantee that the data structure reflects the updates within $\Oo(\log n)$ rounds,
 
 \item and a membership query algorithm that any peer can invoke and  can be answered with efficiency parameter $Q \in \Oo(\log n)$.
   \end{itemize}
 All nodes send and receive at most $\Oo(\Polylog{n})$ messages per round, each comprising at most $\Oo(\Polylog{n})$ bits. Moreover, such a maintenance algorithm is dynamically resource competitive according to Definition~\ref{def:dynamic_resource_comp}. The maintenance protocol ensures that the resilient skip list is maintained effectively for at least $n^d$ rounds with high probability, where $d\geq 1$ can be an arbitrarily chosen constant.
\end{theorem}
\subsubsection[High-level Overview and Technical Contributions]{\textbf{High-level Overview and Technical Contributions}}
Our maintenance protocol (Section~\ref{sec:architecture}) uses a combination of several techniques in a non-trivial way to construct and maintain a churn resilient data structure in $\polylog(n)$ messages per round and $\Oo(\log n)$ rounds.

Our network maintenance protocol is conceptually simple and maintains two networks--the overlay network of peers (called Spartan in Section~\ref{sec:architecture}) and the distributed data structure in which these peers can store data structure information. To ease the description of our distributed algorithm, we think of the overlay and the data structure as two different networks of degree $\Oo(\log n)$. With each peer being a part of both, the overlay and the data structure (see Figure~\ref{fig:high_level_ov}). This way, we can think of our maintenance protocol as a collection of distributed protocols that are running in parallel and are in charge of healing and maintaining these different networks.

The overlay network maintenance protocol is conceptually similar to the ones in~\cite{Gotte_2021,Augustine_2021}. It consists of several phases in which we ensure that the overlay network is robust to an almost linear adversarial churn of $\Oo(n/\log n)$ nodes at each round. Furthermore, the data structure maintenance protocol consists of a continuous \emph{maintenance cycle} in which we quickly perform updates on the distributed data structure despite the high adversarial churn rate. While for maintaining the overlay network we can use the maintenance protocol in~\cite{Augustine_2021} as a black-box, we need to design novel algorithms to maintain the dynamic data structure.

Already existing techniques for skip lists~\cite{Pugh_1990,Aspnes_2007,Goodrich_2006,Jacob_2014} or solving the storage and search problem in the DNC model~\cite{Augustine_2013_b,Attiya_2022} can neither be used nor adapted for our purpose. We need to design novel fast distributed and parallel update protocols that are resilient to high churn rate without compromising the integrity of the data structure. 

\begin{wrapfigure}{r}{0.2\textwidth}
	\includegraphics[width = 0.2\textwidth]{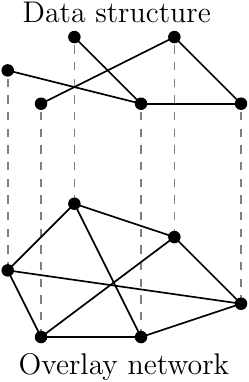}
	\caption{High-level overview.}\label{fig:high_level_ov}
\end{wrapfigure}

Before delving into the protocol's description, we highlight one of the ideas of our paper. As previously mentioned, the network comprises two networks -- the overlay and the distributed data structure. However, instead of having only one network for the current skip list data structure we keep \emph{three}: the first one is a \emph{live} network on which queries are executed, the second one is a \emph{clean} version of the live network on which we perform updates after joins and leaves and the latter, is an additional temporary \emph{buffer} network on which we store the newly added elements. Such a buffer will be promptly \emph{merged} to the clean network during each cycle. This three-network architecture allows us to describe the maintenance process in a clear and simple way.  (Note that the nodes are not actually replicated into multiple copies for each of the networks. The same set of nodes will maintain all the networks logically.)

Another key aspect of our protocol is that when the oblivious adversary removes some nodes from the overlay network (thus from the data structures as well), we have that (1) the overlay network maintenance protocol in~\cite{Augustine_2021} ensures that it will remain connected with high probability 
at every round and (2) all the elements that disappeared from the data structures will be \emph{temporarily replaced} by some surviving group of nodes in the overlay network. Moreover, when (2) happens, all these replacement nodes, in addition to answering their own queries, they also answer the ones addressed to the nodes they are covering for. This is possible since we can always assume some level of redundancy in the overlay network. For example, we can assume that a node $u$ in the overlay network has an updated copy of all its neighbors' values and pointers in the data structure. (Recall that both network are designed to be low-degree networks.)

We now describe a cycle of our maintenance protocol which consists of four major phases. Moreover, we can assume that at the beginning of the cycle, live and clean networks contain the same elements. Next, assume that the adversary replaced $\Oo(n/\log n)$ nodes from the network, this means that $\Oo(n/\log n)$ nodes have been removed from the networks (overlay, live and clean) and new $\Oo(n/\log n)$ have been added to the overlay network. After the churn, there are some (old) nodes in the overlay network covering for the removed ones in both clean and live networks.

\subparagraph*{\textbf{Delete}.}
In the first phase, we promptly remove from the clean network all nodes that have been covered by surviving nodes in the overlay. In Section~\ref{sec:delete}, we describe an $\Oo(\log n)$-round protocol that efficiently refines and updates the clean network.

\subparagraph*{\textbf{Buffer Creation}.}
In the second phase, all the newly added elements are gathered together, sorted and used to create a new temporary skip list, i.e., the buffer network. It is non-trivial to accomplish this under the DNC model in $\Oo(\log n)$ rounds. The first major step in this phase is to efficiently create a sorted list of the newly added elements. To this end we show a technique that first builds a specific type of sorting network~\cite{Ajtai_1983}, creates the sorted list, and, from that, builds the skip list in $\Oo(\log n)$ rounds. This contribution, that can be of independent interest, shows how one can use sorting networks~\cite{Ajtai_1983,Batcher_1968} to efficiently build a skip list despite a high adversarial churn (Section~\ref{sec:creation}). 

\subparagraph*{\textbf{Merge}.}
In the third phase, the buffer network is merged with the clean one. To this end, we propose a novel distributed and parallel algorithm to merge two skip lists together in $\Oo(\log n)$ rounds. Intuitively, the merge protocol can be viewed as a top-down wave of buffer network nodes that is traversing the clean network. Once the wave has fully swept through the clean network, we obtain the merged skip list. All prior protocols for merging skip lists (or skip graphs) took at least $\Oo(k\cdot \log n)$ rounds to merge together two skip lists, where $k$ is the buffer size. In our case the buffer is of size $k=\Oo(n/\log n)$ thus these algorithms would require $\Oo(n)$ rounds to perform such a merge. This is the first distributed and parallel algorithm that merges a skip list of $n$ elements into another one in $\Oo(\log n)$ rounds. This contribution is of independent interest, we believe that the merge procedure will play a key role in extending our work to more general data structures. 

\subparagraph*{\textbf{Update}.}
In the fourth and last phase, we \emph{update} the live network with the clean one by running a $\Oo(1)$ round protocol that applies a \emph{local} rule on each node in the live and clean network (Section~\ref{sec:duplicate}).  

The above maintenance cycle maintains the distributed data structure with probability at least $1-\frac{1}{n^d}$ for some arbitrarily big constant $d\geq 1$. This implies that the expected number of cycles we have to wait before getting the first \emph{failure} is $n^d$ and that the probability that our protocol correctly maintains the data structure for some $r<d$ rounds is at least $1-\frac{1}{n}$.

Beyond skip lists, we show how our techniques generalize to more sophisticated data structures such as skip graphs (Section~\ref{sec:generalize_ds}), and how they can be extended to settings where each node stores multiple elements in the data structure (Section~\ref{sec:generalize_keys}).

\paragraph{Towards a General Framework.}
Our contributions naturally lead to a broader definition: in Section~\ref{sec:general_fw}, we introduce a general framework for distributed data structures under churn and define the class of \emph{$t$-maintainable data structures}, that is, distributed data structures that can be maintained by a continuous maintenance cycle in $\Oo(t)$ rounds (per cycle) with high probability despite a churn rate of $\Oo(n/t)$ per round. In this terminology, our skip list maintenance protocol demonstrates that skip lists and skip graphs are $\Oo(\log n)$-maintainable in the DNC model under near-linear adversarial churn. We believe this framework lays the foundation for a new complexity class of distributed data structures resilient to adversarial churn, enabling future exploration of which classical structures (e.g., trees, graphs) admit efficient $t$-maintenance algorithms, and what inherent trade-offs exist between maintenance time, redundancy, and robustness.

%% file: trunk/architecture.tex
\section{Solution Architecture}\label{sec:architecture}

%
%
Before delving into the description of our maintenance algorithm, we briefly introduce some notation used in what follows. 

\emph{Skip lists}~\cite{Pugh_1990} are randomized data structures organized as a tower of increasingly sparse linked lists. Level $0$ of a skip list is a classical linked list of all nodes in increasing order by key/ID. For each $i$ such that $i>0$, each node in level $i-1$ appears in level $i$ independently with some \emph{fixed} probability $p$. The top lists act as ``express lanes'' that allow the sequence of nodes to be quickly traversed. Searching for a node with a particular key involves searching first in the highest level, and repeatedly dropping down a level whenever it becomes clear that the node is not in the current one. By backtracking on the search path it is possible to shows that no more than $\frac{1}{1-p}$ nodes are searched on average per level, giving an average search time of $\Oo(\log n)$ (see Appendix~\ref{apx:skip_list} for useful properties of randomized skip lists). We refer to the \emph{height} of a skip list $\Ll$ (see Figure~\ref{fig:skip_list}) as the maximum $h$ such that Level $h$ is not empty. 
Given a node $v$ in the skip list $\Ll$ with $n_\Ll$ elements, we use $\neig{\Ll}{\texttt{right}}{\ell}{v}$ and $\neig{\Ll}{\texttt{left}}{\ell}{v}$ to indicate $v$'s right and left neighbors at level $\ell$ respectively. Moreover, when the direction and the levels are not specified we refer to the overall set of neighbors of a node in a skip list, formally we refer to $N_\Ll(v) = \bigcup_{\ell\geq 0} (\neig{\Ll}{\texttt{left}}{\ell}{v}\cup \neig{\Ll}{\texttt{right}}{\ell}{v})$. Furthermore, we define $\ell_{\max}^\Ll(v)$ for a node $v\in \Ll$ to be its maximum height in the skip list\footnote{We omit the superscript when the skip list is clear from the context.} and $\ell_{\max}(\Ll) = \max_{v\in \Ll} \ell_{\max}^\Ll(v)$ to be the overall height of the skip list.

\begin{figure}[htb!]
	\centering
	\includegraphics[scale = 0.5]{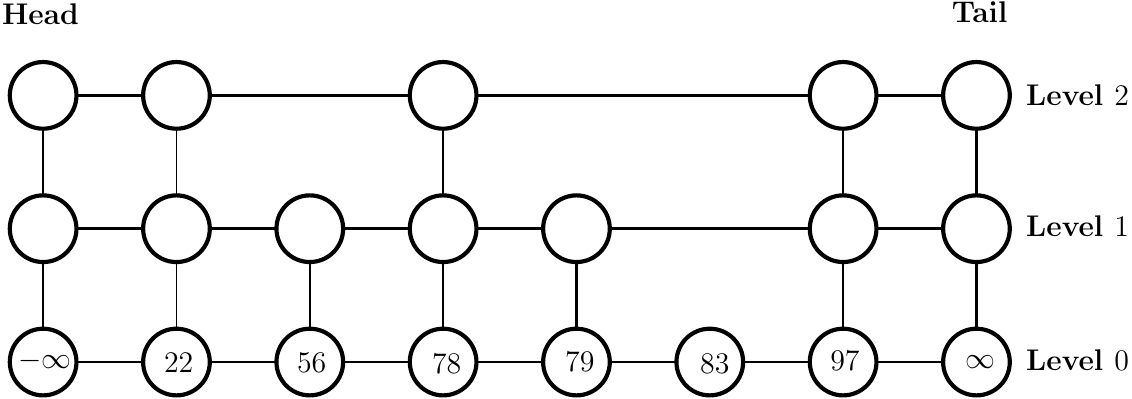}
	\caption{Skip list with $n = 6$ nodes and $3$ levels.}\label{fig:skip_list}
\end{figure}

Our dynamic distributed skip list data structure is architected using multiple ``networks'' (see Figure~\ref{fig:architecture}). Each peer node can participate in more than one network and in some cases more than one location within the same network. We use the following network structures.

\begin{description}
	\item[The \Spartan~Network $\X$] is a wrapped butterfly network that contains all the current nodes. This network can handle heavy churn of up to $\Oo(n/\log n)$ nodes joining and leaving in every round~\cite{Augustine_2021}. However, this network is not capable of handling search queries. 
	\item[Live Network $\Ll$] is the skip list network on which all queries are executed. Some of the nodes in this network may have left. We require such nodes to be temporarily represented by their replacement nodes (from their respective neighbors in $\X$).
	\item[Buffer Network $\B$] is a skip list network on which we maintain all new nodes that joined recently.
	\item[Clean Network $\C$] is a skip list network that seeks to maintain an updated version of the data structure that includes the nodes in the system. 
\end{description}
\begin{figure}[htb!]
	\centering
	\includegraphics[scale=0.7]{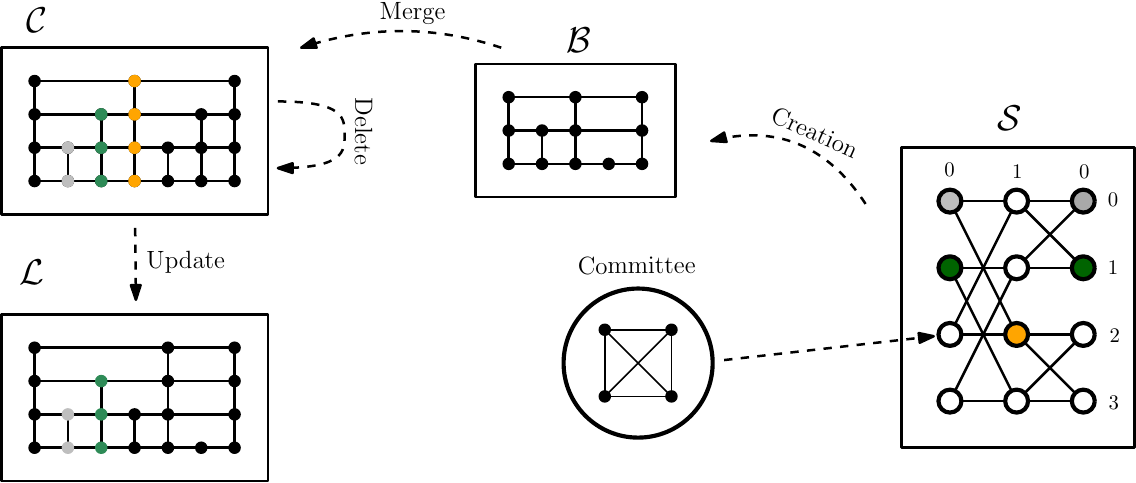}
	\caption{Schematic representation of the architecture and the Maintenance cycle described in Algorithm~\ref{algo:overview}. Colored nodes in $\Ll$ and $\C$ are nodes that have been removed by the adversary and that are being covered by some committee of nodes (of the same color) in $\X$.}\label{fig:architecture}
\end{figure}
Moreover, in all skip lists, when a node exits the system, it is operated by a selected group of nodes. In the course of our algorithm description, if a node $u$ is required to perform some operation, but is no longer in the system, then its replacement node(s) will perform that operation on its behalf. Note further that some of the replacement nodes themselves may need to be replaced. Such replacement nodes will continue to represent $u$. The protocol assumes a short ($\Theta(\log n)$ round) initial ``bootstrap'' phase, where there is no churn\footnote{Without a bootstrap phase, it is easy to show that the adversary can partition the network into large pieces, with no chance of forming even a connected graph.} and it initializes the underlying network. More precisely, the bootstrap is divided in two sub-phases in which we (i)  build the underlying churn resilient network described in Section~\ref{sec:churn_resilient_network} in $\Oo(\log n)$ rounds and, (ii) we build the skip lists data structures $\Ll$ and $\C$ (initially $\Ll = \C$) using the $\Oo(\log n)$ rounds technique described in Section~\ref{sec:creation}; after this, the adversary is free to exercise its power to add or delete nodes up to the churn limit and the network will undergo a continuous maintenance process.
The overall maintenance of the dynamic distributed data structure goes through cycles. Each cycle $c \geq 1$ in Algorithm~\ref{algo:overview} consists of four sequential phases. Without loss of generality, assume that initially $\Ll$ and $\C$ are the same. We use the notation $\LL{c}$ (resp. $\XX{c},\CC{c},\BB{c}$) to indicate the network $\Ll$ (resp. $\X,\C, \B$) during the cycle $c\in\mathbb{N}$, when it is clear from the context we omit the superscript to maintain a cleaner exposition.
\input{pseudocodes/overview}

\subsection{The \Spartan~network}\label{sec:churn_resilient_network}

A  widely used approach to maintain a stable overlay network under high adversarial churn is to organize nodes into small-sized
\emph{committees}~\cite{Augustine_2013_b,Augustine_2015_a,Drees_2016,Augustine_2021}. A committee is a clique of small ($\Theta(\log n)$) size composed of essentially ``random'' nodes. A committee can be efficiently constructed, and more importantly, \emph{maintained} under heavy churn. Maintenance must ensure that each committee is made up of a random set of $O(\log n)$ peers, thus making it difficult for an oblivious adversary to remove all nodes from a committee. Since each committee behaves as a single unit, the loss of a (strict) subset of peers, which an oblivious adversary can inflict,  will not disrupt the committee; the remaining peers can ensure that the continues to operate as a single unit.  Moreover, a committee can be used to ``delegate'' nodes to perform any kind of operation. 

In this work, we build on the Spartan network design~\cite{Augustine_2021}, which improves on earlier committee-based approaches by (a) tolerating an adversarial churn rate of $\Oo(n/\log n)$ per round, and (b) achieving a construction time of $\Oo(\log n)$ rounds with high probability.
Each committee is a dynamic random clique of size $\Theta(\log n)$ in which member nodes change continuously with the guarantee that the committee has $\Theta(\log n)$ nodes as its members at any given time. These committees are arranged into a 
\emph{wrapped butterfly} network~\cite{Leighton_2014,Mitzenmacher_2017}. The wrapped butterfly has $2^k$ rows and $k$ columns such that $k2^k\in\Oo(n/\log n)$ nodes and edges. The nodes correspond to committees and are represented by pairs $(r,\ell)$ where $\ell$ is the \emph{level} or \emph{dimension} of the committee ($0\leq \ell\leq k$) and $r$ is a $k$-bit binary number that denotes the \emph{row} of the committee. Two committees $(r,\ell)$ and $(r',\ell')$ are linked by an edge that encodes a complete bipartite graph if and only if $\ell' = \ell+1$ and either: (1) both committees are in the same row i.e., $r=r'$, or (2) $r$ and $r'$ differ in precisely the $(\ell+1 (\text{mod } k))$th bit in their binary representation. Finally, the first and last levels of such network correspond to the same level (i.e., the butterfly is \emph{wrapped}). In particular, committees $(r,0)$, and $(r,k)$ are the same committees. An alternative way to represent this structure is to create edges from the $k-1$th level to the level $0$\footnote{We decided to avoid this representation, to maintain a clear exposition and easy to read pictures.}. 
Figure~\ref{fig:architecture} shows an example of \Spartan~network with $4$ rows and $2$ columns in which every node of the two-dimensional wrapped butterfly encodes a committee of $\Theta(\log n)$ random nodes and each edge between two committees encodes a complete bipartite graph connecting the vertices among committees. Notice that the first and last columns are the same set of committees since the butterfly is wrapped. 

The Spartan network is constructed during an initial bootstrap phase, assuming no adversarial activity\footnote{Having such phase is a common assumption in the DNC model~\cite{Augustine_2013_b,Augustine_2013_a,Augustine_2015_a,Augustine_2015_b,Augustine_2012,Augustine_2021}}, as follows: (1) a random leader node is elected; (2) a binary tree of height $\Oo(\log n)$ is built; (3) the first $N = k2^k$ nodes are organized into a cycle, and designated as committee leaders; (4) the cycle is transformed into a wrapped butterfly structure; (5) the remaining nodes are assigned randomly to committees, which then internally form cliques; (6) committee leaders exchange node ID lists to establish the bipartite overlay edges between committees.


After the bootstrap, Spartan maintains stability using a continuous maintenance cycle that runs every $\Theta(\log \log n)$ rounds. In each cycle, all nodes are reassigned uniformly at random to new committees, and any newly joined nodes are placed randomly as well. While each node can locally detect when a neighbor leaves (via the loss of the direct link), the system’s stability does not depend on explicitly tracking individual departures. Instead, the maintenance cycle ensures that, with high probability, each committee maintains $\Theta(\log n)$ members and the network remains connected, as long as the churn rate is $\Oo(n/\log n)$.

As a contribution of this work, we generalize the Spartan framework to handle variable network sizes, removing the assumption that the number of nodes remains fixed at $n$. We provide a reshaping protocol (Appendix~\ref{apx:reshaping_protocol}) that dynamically adjusts the dimensionality of the wrapped butterfly as the network size increases or decreases, allowing the system to maintain its structural and robustness guarantees under highly dynamic and more general conditions. 
\subsection{Replacing nodes that have been removed by the adversary}\label{sec:covering_for_nodes}
We describe how to maintain the live network $\Ll$ and the clean network $\C$ when churn occurs. Assume that node $v\in V$ left the network. To preserve $\C$ and $\Ll$ structures, we require its committee in $\X$ ``to cover'' for the disappeared node. In other words, when node $v$ leaves the network, its committee members in $\X$ will take care of all the operations involving $v$ in $\Ll$ and $\C$. This temporarily preserves the structure of these networks, as if $v$ were still present. Moreover, we assume that at every round in the \Spartan~network, all the nodes within the same committee communicate their states (in $\X,\Ll$, and $\C$) to all the other committee members, i.e., 
each node $u$ in a committee knows all neighbors $N_\X(u)$ and all neighbors of these nodes in $\Ll$ and $\C$, and the committees they belong to\footnote{The expected space needed to store the committee's coordinates of $v$' neighbors in $\Ll$ and $\C$ is $\Oo(\log n)$ because every node in a skip list has expected degree of $\Oo(\log n)$, \whp}. Furthermore, when a node $v$ leaves the network, each node in its committee in $\X$ will connect to $v$'s neighbors in $\Ll$ and $\C$. This can be done in constant time and will increase the degree of the nodes in the data structures to $\Oo(\polylog(n))$. Notice that each committee in $\X$ is guaranteed to be robust against a $\Oo(n/\log n)$ churn rate, i.e., the probability that there exists some committee of size $o(\log n)$ is at most $1/n^d$ for an arbitrarily large $d\geq 1$ (see Theorem 6 in~\cite{Augustine_conference_18}). This implies that every committee of nodes in $\X$ is able to cover for the removed nodes in $\Ll$ and $\C$ with high probability. In particular, since each committee survives the churn with high probability (as guaranteed by Theorem 6~~\cite{Augustine_conference_18}), it cannot be fully destroyed. This in turn implies that the committee will succeed in covering the deleted node with high probability, ensuring the continuity of $\Ll$ and $\C$ despite adversarial removals. Figure~\ref{fig:committee_replacement} shows an example of how the committees in the Spartan network cover for its members that were removed by the adversary. More precisely, in Figure~\ref{fig:replacement_after} we show how the orange node's committee is covering for it after being dropped by the adversary. An alternative way to see this is to assume that the committee creates a temporary virtual red node in place of the orange one (see Figure~\ref{fig:replacement_equivalent}). This virtual node represents the collective responsibility of the committee, effectively simulating $v$’s presence.

 \begin{figure}[htb!]
 	\captionsetup[subfigure]{justification=centering}
 	
 	\begin{subfigure}{0.3\textwidth}
 		
 		\raisebox{-1.35\height}{\includegraphics[scale=0.5]{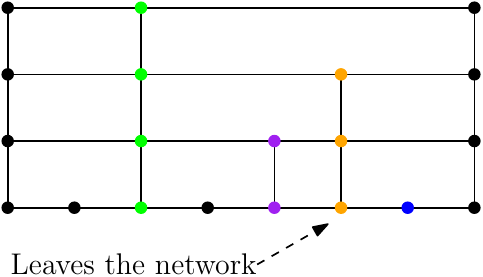}}
 		\caption{}\label{fig:replacement_before}
 	\end{subfigure}
 	\begin{subfigure}{0.3\textwidth}
 		\includegraphics[scale=0.5]{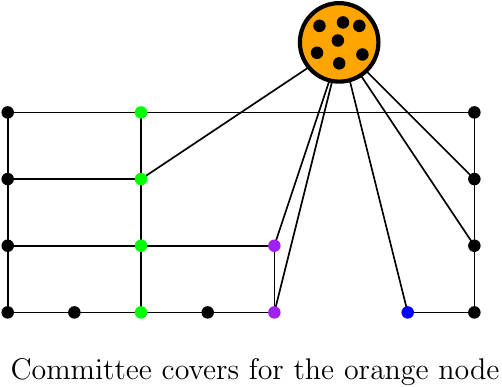}
 		\caption{}\label{fig:replacement_after}
 	\end{subfigure}
 	\begin{subfigure}{0.3\textwidth}
 		\includegraphics[scale=0.5]{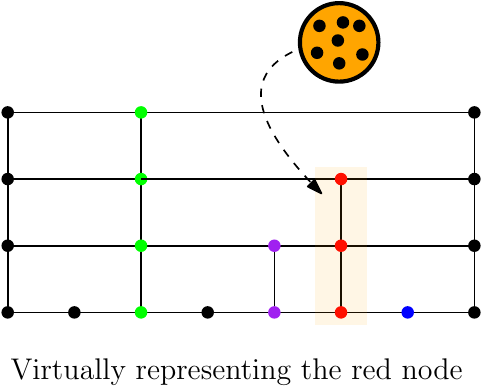}
 		\caption{}\label{fig:replacement_equivalent}
 	\end{subfigure}
 	\caption{ In Figure~\ref{fig:replacement_before} the orange node gets removed by the adversary from the overlay network. Figure~\ref{fig:replacement_after} shows how all its committee neighbors in the overlay cover for the orange node. Figure~\ref{fig:replacement_equivalent} is an equivalent visualization of Figure~\ref{fig:replacement_after}, the committee replaces the orange node with a ``virtual'' red node. This will temporarily preserve the skip list structure of Figure~\ref{fig:replacement_before}.}\label{fig:committee_replacement}
 \end{figure}

\begin{lemma}\label{lemma:taking_over}
    With high probability, replacing nodes removed by the adversary requires $\Oo(1)$ rounds and work proportional (within $\polylog(n)$ factors) to the number of nodes removed.
\end{lemma}
\begin{proof}
When a node $v$ leaves the network, its committee members in $\X$ take over its connections in $\Ll$ and $\C$ by creating direct edges to $v$'s neighbors in these structures. This takeover operation can be completed in $\Oo(1)$ rounds, as each committee member already knows the necessary local neighborhood information.

Each committee has expected size $\Oo(\log n)$, and each skip list node has expected degree $\Oo(\log n)$, so the total number of new edges created by $v$’s committee is at most $\Oo(\log^2 n)$. Given that the system can experience up to $\Oo(n / \log n)$ churn events in one $\Oo(\log n)$-round cycle, the overall work is $\Oo(n \cdot \polylog(n))$, satisfying dynamic resource competitiveness.

Crucially, by Theorem 6 in~\cite{Augustine_conference_18}, each committee survives churn with high probability; that is, the probability that a committee’s size drops below $\Oo(\log n)$ is at most $1/n^d$ for arbitrarily large $d \geq 1$. Therefore, with high probability, the committee is able to successfully cover the departed node, ensuring the continuity of $\Ll$ and $\C$ despite adversarial removals.
\end{proof}
\input{trunk/delete}

\input{trunk/creation}

\input{trunk/merge_phase}
\input{trunk/update}

%% file: pseudocodes/overview.tex
\begin{algorithm}[htb!]
	\caption{Overview of the distributed skip list maintenance process. }
	\label{algo:overview} 
	
	Starting from a skip list $\LL{1}=\CC{1}$
	
	\ForEach{\textbf{Cycle} $c\geq 1$}{
	
		\textbf{Phase 1 (Deletion):} All the replacement nodes in $\CC{c}$ are removed. Note that nodes that leave the system during the replacement process may remain in $\CC{c}$.
		
		\textbf{Phase 2 (Buffer Creation):} All the new nodes that were churned into the system since Phase 2 of the previous cycle join together to form $\BB{c}$.
		
		\textbf{Phase 3 (Merge):} The buffer $\BB{c}$ created in Phase $2$ is merged into $\CC{c}$, i.e., $\CC{c+1} \leftarrow \CC{c} \cup \BB{c}$. 
		
		\textbf{Phase 4 (Update):} Update the live network $\LL{c}$ with the clean network $\CC{c}$, i.e., $\LL{c+1} \leftarrow \CC{c+1}$.

	} 	
\end{algorithm}

%% file: trunk/delete.tex
\subsection{Deletion}\label{sec:delete}
We describe how to safely clean the network $\C$ in $\Oo(\log n)$ rounds from the removed nodes for which committees in the \Spartan~network $\X$ are temporarily covering. The elements in the clean network $\C$ before the cleaning procedure can be of two types: (1) not removed yet and (2) patched-up, i.e., nodes that have been removed by the adversary and for which another group of nodes in the network is covering. 
A trivial approach to efficiently remove the patched-up nodes from $\C$, would be the one of destroying the clean network $\C$ by removing these nodes and then reconstructing the cleaned skip list from scratch. However, this approach does not satisfy our dynamic resource competitiveness constraint (see Definition~\ref{def:dynamic_resource_comp} and Section~\ref{sec:contributions}); the effort we must pay to build the skip list from scratch may be much higher than the adversary's total cost\footnote{The adversary's cost is the churn applied to the network, i.e., the number of nodes removed in a certain time period.}. Hence, we show an approach to clean $\C$ that guarantees dynamic resource competitiveness. For the sake of explanation, assume that the patched-up nodes are red and the one that should remain in $\C$ are black. The protocol is executed in all the levels of the skip list in parallel. The high-level idea of the protocol is to create a tree rooted in 
the skip list's left-topmost sentinel by backtracking from some ``special'' leaves located on each level of the skip list. Once the rooted tree is constructed, another backtracking phase in which we build new edges connecting black nodes separated by a chain of red nodes is executed. Figure~\ref{fig:delete_logn} shows an example of the delete routine (Algorithm~\ref{algo:delete}) run on all the levels of the skip list in Figure~\ref{fig:delete_base}. Moreover, Figure~\ref{fig:delete_0} and Figure~\ref{fig:delete_1} show the execution of the delete routine at level $0$. First, the tree rooted in the left-topmost sentinel is built by backtracking from the leaves (i.e., from the black nodes that have at least one red neighbor). Next, the tree is traversed again in a bottom-up fashion and green edges between black nodes separated by a list of red nodes are created (Figure~\ref{fig:delete_2}). Figure~\ref{fig:delete_3} depicts the skip list after running Algorithm~\ref{algo:delete} in parallel on every level. Finally, Figure~\ref{fig:delete_4} shows the skip list after the deletion phase.

\begin{algorithm}[htb!]
	\caption{Delete routine for level $\ell$. See Figure~\ref{fig:delete_logn} for an illustration.}
	\label{algo:delete}

    \nonl \textbf{Tree Formation Phase ($\Oo(\log n)$ rounds):}

    Let $B$ be the set of black nodes at level $\ell$ with at least one red neighbor. In this step, we are required to form the shortest path tree $T$ rooted at the left-topmost sentinel with $B$ as its leaves. A simple bottom-up approach can be used to build this tree $T$ in $O(\log n)$ rounds.

    \nonl~

    \nonl \textbf{Message Initialization ($(\Oo(1)$ rounds):}

    Every $b \in B$ creates a message containing a pair of items as follows.
    {\begin{description}
        \item[1.] $(\dot{b},\dot{b})$ if both its neighbors are red.
        \item[2.] $(\dot{b},b)$ if its left neighbor is red and its right neighbor is black
        \item[3.] $(b,\dot{b})$ if its left neighbor is black and its right neighbor is red
    \end{description}}
    Each $b$ then sends its message to its parent in $T$.
    
    \nonl~
    
    \nonl \textbf{Propagation Phase ($\Oo(\log n)$ rounds):}

    This is executed by each node $u$ on the tree $T$. Node $u$ waits to hear messages from all its children. If it has only one child, it will propagate the message to its parent in the tree. Otherwise (i.e., it has two children) it receives two messages: $(w,x)$ from its child from below and $(y, z)$ from its child from its right. We will show later that either $x$ and $y$ are both dotted or neither of them are dotted. 

    \If{both $x$ and $y$ are dotted}{
        Introduce $x$ to $y$ so that they can form an edge between them at level $\ell$.
    }
    Node $u$ then propagates the message $(w,z)$ to its parent on the tree.

\end{algorithm}

\begin{figure}[htb!]
 \captionsetup[subfigure]{justification=centering}
	\centering
    \begin{subfigure}{0.47\textwidth}
		\includegraphics[scale=0.5]{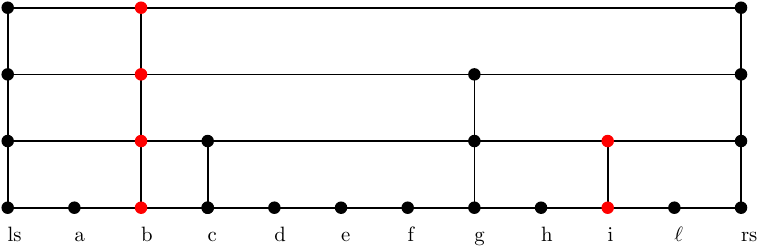}
		\caption{}\label{fig:delete_base}
	\end{subfigure}
	\begin{subfigure}{0.45\textwidth}
		\includegraphics[scale=0.5]{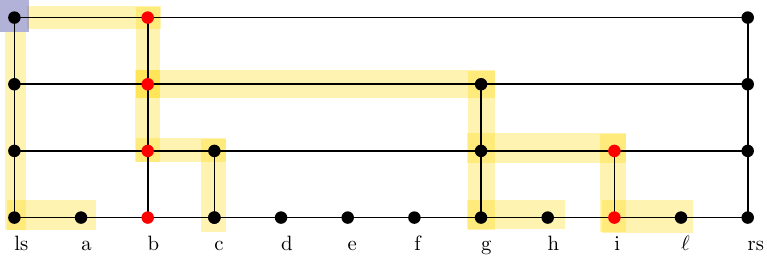}
    	\caption{}\label{fig:delete_0}
   \end{subfigure}
 
	\begin{subfigure}{0.47\textwidth}
		\includegraphics[scale=0.5]{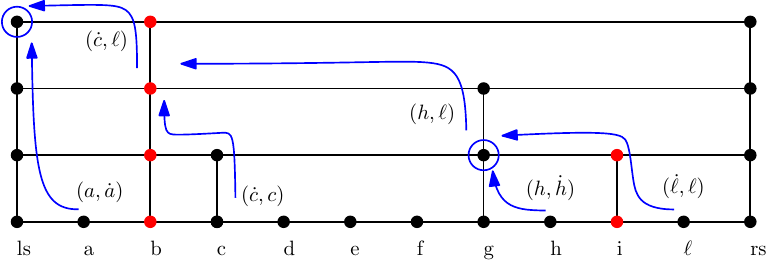}
		\caption{}\label{fig:delete_1}
	\end{subfigure}
	\begin{subfigure}{0.45\textwidth}
		\includegraphics[scale=0.5]{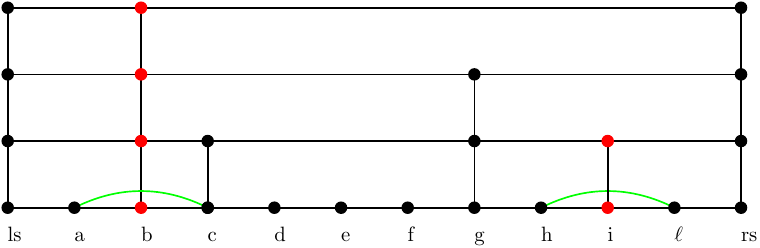}
    	\caption{}\label{fig:delete_2}
   \end{subfigure}

   \begin{subfigure}{0.47\textwidth}
		\includegraphics[scale=0.5]{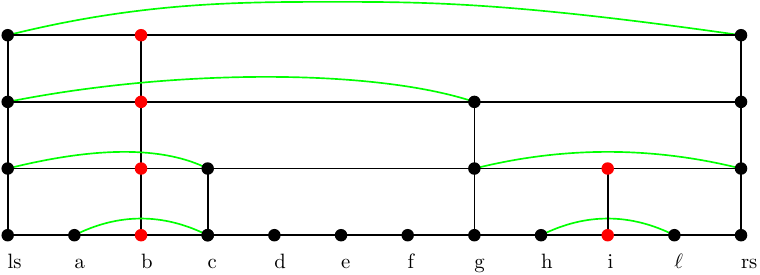}
		\caption{}\label{fig:delete_3}
	\end{subfigure}
	\begin{subfigure}{0.45\textwidth}
		\includegraphics[scale=0.5]{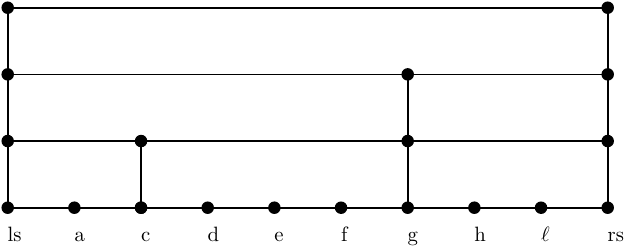}
    	\caption{}\label{fig:delete_4}
   \end{subfigure}

	\caption{Example of the delete phase on a skip list. Red nodes are the one that must be deleted from the skip list, blue arrow indicates the path towards the left topmost sentinel in the skip list from the black nodes. Figure~\ref{fig:delete_base} is the skip list before the deletion phase. Figure~\ref{fig:delete_0} shows the tree rooted in the left-topmost sentinel created by the Tree Formation Phase. Figure~\ref{fig:delete_1} shows the Message Initialization and the Propagation Phase executed at the bottom most level of the skip list where blue circles indicate that an edge between two black nodes have been created. Figure~\ref{fig:delete_2} shows the skip list after the bottom most lever run Algorithm~\ref{algo:delete}. Figure~\ref{fig:delete_3} is the skip list after running Algorithm~\ref{algo:delete} at all its levels. Finally, Figure~\ref{fig:delete_4} is the skip list after the deletion phase. The left and right sentinels are represented as ls and rs, respectively.}\label{fig:delete_logn}
 \end{figure}

\begin{lemma}\label{lemma:delete}
    Algorithm~\ref{algo:delete} executed on a skip list at all levels $0 \le \ell \in \Oo(\log n)$ correctly creates an edge between every pair  of black nodes that are separated by a set of red nodes at each level $\ell$ in $\Oo(\log n)$ rounds \whp. Moreover, the total work performed by this procedure is within a $\polylog(n)$ factor of the number of red nodes.
\end{lemma}
\begin{proof}

We first establish a simple but useful invariant that holds for all nodes $u$ in $T$. Consider a pair $(x,y)$ propagated upwards by a node $u$ to its parent. We note that $x$ (resp., $y$) (in its un-dotted form) is the leftmost (resp., rightmost) leaf in the subtree rooted at $u$. This can be shown by induction. Clearly the statement is true if $u \in B$ as it creates the message to be of the form $(u,u)$ (either of them possibly dotted). The inductive step holds for $u$ at higher levels by the manner in which messages from children are combined at $u$. 

Note that the procedure  forms edges between pairs of nodes $(b, b')$ at some level $\ell$ only if they are both in $B$ (corresponding to the execution pertaining to level $\ell$). Thus, for correctness, it suffices to show that the edge is formed if and only if all the nodes between $b$ and $b'$ at level $\ell$ are red. 

To establish the forward direction of the bidirectional statement, let us consider the formation of an edge $(b,b')$ at some node $u$ that is at a level $\ell' \ge \ell$. It coalesced two pairs of the form $(x,\dot{b})$ and $(\dot{b'}, y)$. The implication is that $b$ has a right neighbor at level $\ell$ that is red and $b'$ has a left neighbor at level $\ell$ that is red. Moreover, they are the rightmost leaf and left most leaf respectively of the subtrees rooted at $u$'s children. Thus, none of the nodes between $b$ and $b'$ at level $\ell$ are in $T$. This implies that they must all be red nodes.

To establish the reverse direction, let us consider two nodes $b$ and $b'$ at level $\ell$ that are both black with all nodes between them at level $\ell$ being red. Consider their lowest common ancestor $u$ in $T$. Clearly, $u$ must have two children, say, $v$ and $w$. From the invariant established earlier, the message sent by $v$ (resp., $w$) must be of the form $(\ast, \dot{b})$ (resp., $(\dot{b'}, \ast)$). Thus, $u$ must introduce $b$ and $b'$, thereby forming the edge between them.

Since $T$ is the shortest path rooted at the left-topmost sentinel, its height is at most $\Oo(\log n)$ \whp~Both the tree formation and propagation phases are bottom-up procedures that take time proportional to the height of $T$. Thus, the total running time is $\Oo(\log n)$. 

Finally, each node in $T$ sends at most $\Oo(1)$ messages up to its parent in $T$. This implies that the number of messages sent is proportional (within $\polylog(n)$ factors) to the number of leaves in $T$. Also, all edges added in the tree are between pairs of nodes $(b, b')$ in $B$ that are consecutive at level $\ell$ (i.e., no other node in $B$ lies between $b$ and $b'$ at level $\ell$). Thus, work efficiency is established. 
\end{proof}

%% file: trunk/creation.tex
\subsection{Buffer Creation}\label{sec:creation}
To construct the Buffer Network $\B$, we begin with a set of nodes that recently joined the system due to churn. These nodes are initially unsorted and must be organized into a skip list structure. Here, the core challenge is to sort them quickly and reliably in a fully distributed setting. The current best-known approach to sort $n$ elements using $\Polylog{n}$ incoming and outgoing messages for each node during each round requires $\Oo(\log^3 n)$ time~\cite{Augustine_2022}. To improve on such result, we rely on sorting networks theory (e.g.~\cite{Batcher_1968,Ajtai_1983,Leighton_2014}). A sorting network is a fixed graph-like structure designed specifically for sorting any input permutation efficiently.
More precisely, it can be viewed as a circuit-looking directed acyclic graph with $n$ inputs and $n$ outputs in which each node has exactly two inputs and two outputs (i.e., two incoming edges and two outgoing edges). The key property of a sorting network is that, for any input permutation, the outputs are guaranteed to be sorted. A sorting network is characterized by two parameters: the number of nodes resp. \emph{size} and the \emph{depth}. The space required to store a sorting network is determined by its size while the overall running time of the sorting algorithm is determined by its depth. Figure~\ref{fig:sorting_net} shows an example of sorting network on $4$ elements.
\begin{figure}[htb!]
	\captionsetup[subfigure]{justification=centering}
	\centering
	\begin{subfigure}{0.15\textwidth}
			\centering
		\includegraphics[scale=0.4]{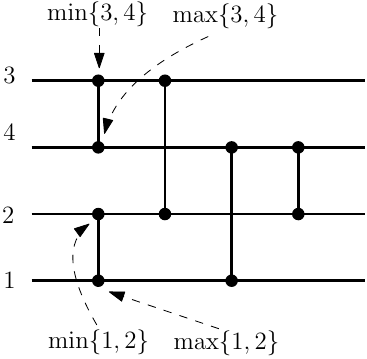}
		\caption{}\label{fig:sorting_net_a}
	\end{subfigure}
	\begin{subfigure}{0.3\textwidth}
			\centering
		\includegraphics[scale=0.4]{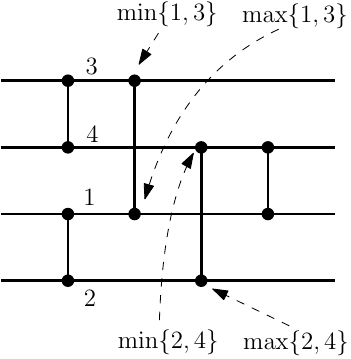}
		\caption{}\label{fig:sorting_net_b}
	\end{subfigure}
    \begin{subfigure}{0.2\textwidth}
			\centering
		\includegraphics[scale=0.4]{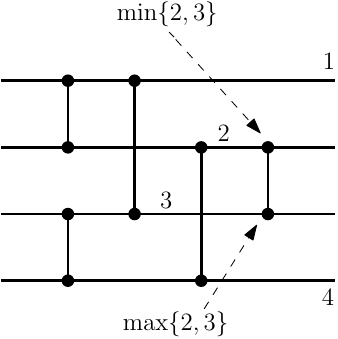}
		\caption{}\label{fig:sorting_net_c}
	\end{subfigure}
    \begin{subfigure}{0.3\textwidth}
			\centering
		\includegraphics[scale=0.4]{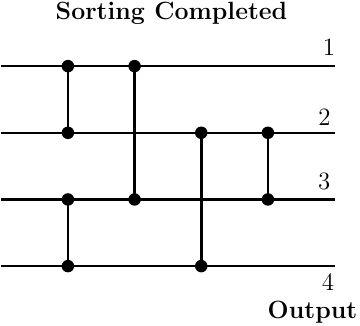}
		\caption{}\label{fig:sorting_net_d}
	\end{subfigure}
	\caption{Example of a sorting network~\cite{Batcher_1968} on $4$ numbers. Every crossing edge connecting two nodes in different levels is a node in the DAG representation. The upper output of a crossing edge connecting two nodes in different levels is the minimum between the two inputs, while the lower output is the maximum. Execution of the sorting network is left to right.}\label{fig:sorting_net}
\end{figure}

Here we use the AKS network~\cite{Ajtai_1983} that allows to sort $n$ elements in $\Oo(\log n)$ rounds. 

To use sorting networks in our distributed setting, we must simulate their structure on a suitable topology. In our case, this is achieved by embedding the AKS network onto a twin-butterfly, which we construct on top of the Spartan network.

Indeed, it is well known that sorting networks can be simulated on structured topologies such as a butterfly networks or other hypercubic networks (see Chapter 3.5 in~\cite{Leighton_2014}). To leverage this in our setting, we aim to embed the AKS sorting network onto a butterfly-like network. For this purpose, we use a result by Maggs and V{\"{o}}cking~\cite{Maggs_2000} which shows that an AKS network can be embedded into a \emph{multibutterfly}~\cite{Upfal_1992}, a generalization of the classic butterfly topology.

A $k$-dimensional multibutterfly network consists of $k+1$ levels, each consisting of $2^k$ nodes. Nodes at level $\ell$ are labeled as $(\ell,i)$ for $0\leq\ell\leq k$ and $0\leq i< 2^k$. Each level is partitioned into $2^k$ groups
\begin{align*}
A_{\ell,j} = \left\{(\ell,i):\left\lfloor\frac{i}{2^{k-\ell}}\right\rfloor = j\right\}\quad\text{ for }0\leq j<2^k
\end{align*}
Nodes in group $A_{\ell,j}$ are connected to those in $A_{\ell+1,2j}$ and $A_{\ell+1,2j+1}$, creating a butterfly-like network. 

To embed the AKS network, we focus on a specific subclass of multibutterfly networks called \emph{k-folded butterflies}. These are obtained by superimposing $k$ butterfly networks and ``folding'' them into one. That is, the nodes from each butterfly with the same label are merged together, and the edges are colored so that each color class forms a valid butterfly. This folding reduces the number of physical nodes while preserving the logical connectivity of multiple butterflies. 

Formally, suppose the edges of a $d$-dimensional multibutterfly of degree $4k$ can be colored by $k$ colors such that the network induced by the edges of each color are isomorphic to the $d$-dimensional butterfly. Then this multibutterfly can be constructed by folding $k$ butterfly networks. By folding, we mean that the labels of the nodes in the $A_{\ell,j}$ sets of each of these butterflies are permuted and the $k$ nodes with the same label in distinct butterflies are merged together to form a multibutterfly node\footnote{Intuitively, a $k$-folded butterfly overlays $k$ butterfly networks and merges nodes with the same label, allowing multiple logical networks to coexist on a single physical graph.}.

In particular, Maggs and V{\"{o}}cking~\cite{Maggs_2000} show how to embed the AKS network into a \emph{twin-butterfly}, which is the case where two butterfly networks are folded together (i.e., a $2$-folded butterfly). In such a network, each leaf node has degree $4$, and each internal node has degree $8$.
For completeness we state the result we are interested in.
\begin{theorem}[Theorem 4.1 in~\cite{Maggs_2000}]\label{thm:aks_embedding}
	An AKS network $\A$ of size $N =\Oo(n\log n)$ can be embedded into a twinbutterfly $\M$ of size at most $\kappa\cdot N+o(N)$ with load $1$, dilation $2$, and congestion $1$, where $\kappa\leq 1.352$ is a small constant depending on the AKS parameters.
\end{theorem}
Here the \emph{load} of an embedding is the maximum number of nodes of the AKS network mapped to any node of the twinbutterfly. The \emph{congestion} is the maximum number of paths that use any edge in the twinbutterfly, and the \emph{dilation} is the length of the longest path in the embedding.
\begin{lemma}\label{lemma:constr_twinb}
	There exists a distributed algorithm that builds a twinbutterfly network $\M$ in $\Oo(\log n)$ time.
\end{lemma}
\begin{proof}
	Let $n_\B = |C|$ be the number of nodes that must join together to form the Buffer network $\B$. The temporary twinbutterfly $\M$ can be built in $\Oo(\log n)$ rounds by adapting the distributed algorithm for building a butterfly network in~\cite{Augustine_2021} described in Section~\ref{sec:churn_resilient_network}: (1) a leader $l$ is elected in $\X$; (2) a $\Oo(\log n)$ height binary tree rooted in $l$ is constructed; (3) a cycle using the first $n_\B (\log n_\B +1)$ nodes of the tree in-order traversal is built; and, (4) the cycle is then transformed into the desired twinbutterfly network with $\log n_\B+1$ columns and $n_\B$ rows in $\Oo(\log n)$ rounds.
\end{proof}
The idea is to use such a butterfly construction algorithm to build a twinbutterfly $\M$ in $\Oo(\log n)$ rounds and simulate the AKS algorithm on it. 
Moreover, $\M$ being a data structure build on top of the \Spartan~network, it naturally enjoys a churn resiliency property. In fact, every time a node $v$ is deleted from $\X$, a committee will temporary cover for $v$ in $\M$ (see Section~\ref{sec:covering_for_nodes}). 
Finally, we show that the $\B$ can be created in $\Oo(\log n)$ rounds.
\begin{lemma}\label{lemma:creation}
	The Creation phase can be computed in $\Oo(\log n)$ rounds \whp.
\end{lemma}
\begin{proof}
	$\M$ can be built in $\Oo(\log n)$ rounds (Lemma~\ref{lemma:constr_twinb}). Moreover, to handle cases in which $\kappa \cdot N+ o(N)$ (see Theorem~\ref{thm:aks_embedding}) is greater than the number of nodes in the Spartan network $\X$, we allow for nodes to represent a constant number of temporary ``dummy'' nodes for the sake of building the needed multibutterfly. Once $\M$ has been constructed, we run AKS on such network and after $\Oo(\log n_\B) =\Oo(\log n)$ rounds we obtain the sorted list at the base of the new skip list we want to build. Next, each node $u$ in the buffer computes its maximum height $\ell_{\max}^\B(u)$ in the buffer skip list $\B$ and each level $0<\ell\leq \ell_{\max}(\B)$ is created by copying the base level of the skip list. Moreover, each node $u$ at level $0<\ell\leq \ell_{\max}(\B)$ can be of two types, \emph{effective} or \emph{fill-in}. We say that a node $u$ is \emph{effective} at level $\ell$ if $\ell\leq \ell_{\max}^\B(u)$, and \emph{fill-in} otherwise (see Figure~\ref{fig:fill_in_1}). Next, we use the same parallel $\Oo(\log n)$ rounds rewiring technique used in the Deletion phase (see Section~\ref{sec:delete}) to exclude fill-in nodes from each level and obtain a skip list of effective nodes (see Figures~\ref{fig:fill_in_2}-\ref{fig:fill_in_3}). Wrapping up, to create the buffer network we need $\Oo(\log n)$ rounds to build the twibutterfly $\M$, $\Oo(\log n)$ rounds to run AKS on $\M$, $\Oo(\log n)$ rounds \whp~to duplicate the levels (see Lemma~\ref{lemma:exp_hight} in Appendix~\ref{apx:skip_list}) and $\Oo(\log n)$ rounds to run the rewiring procedure. Thus, we need $\Oo(\log n)$ rounds
	\whp
\end{proof}
\begin{figure}[htb!]
	\captionsetup[subfigure]{justification=centering}
	\centering
	\begin{subfigure}{0.3\textwidth}
		\centering
		\includegraphics[scale=0.55]{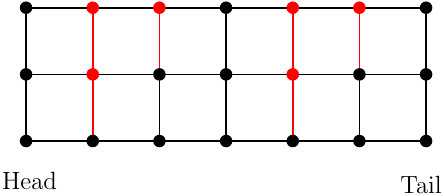}
		\caption{}\label{fig:fill_in_1}
	\end{subfigure}
	\begin{subfigure}{0.3\textwidth}
		\centering
		\includegraphics[scale=0.55]{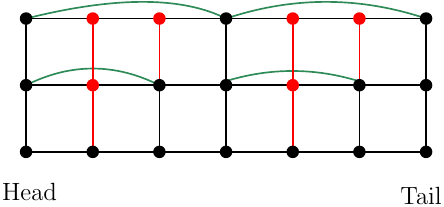}
		\caption{}\label{fig:fill_in_2}
	\end{subfigure}
	\begin{subfigure}{0.3\textwidth}
		\centering
		\includegraphics[scale=0.55]{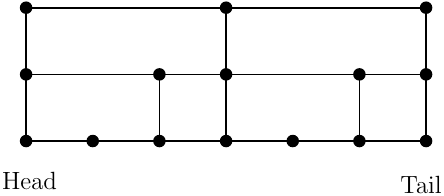}
		\caption{}\label{fig:fill_in_3}
	\end{subfigure}	
	\caption{Example of create procedure. In Figure~\ref{fig:fill_in_1}, freshly created Buffer network in which the black nodes are the effective nodes and red ones are the fill-in nodes. Figure~\ref{fig:fill_in_2} shows the buffer network after the parallel rewiring phase in which the green edges are the one created by the delete procedure. Finally, Figure~\ref{fig:fill_in_3} shows the buffer network after the Creation phase.}\label{fig:create_fast}
\end{figure}
This completes the creation of the Buffer network $\B$, which is now a properly structured skip list ready to be merged into the clean network $\C$.

%% file: trunk/merge_phase.tex
\subsection{Merge}\label{sec:merge}
Once the Buffer Network $\B$ is constructed, the next step is to merge it into the Clean Network $\C$, updating the distributed skip list to include the newly joined nodes. This operation must be carried out in a fully distributed manner and under high churn.

To this end, we introduce the \textsc{WAVE} protocol, a distributed merge algorithm
based on a top-down traversal of $\C$ by groups of nodes in $\B$. The merge proceeds in rounds, with buffer nodes moving in cohesive groups, that are set of adjacent nodes at the same level in $\B$ that move and act together. These groups progressively insert themselves into the skip list $\C$, level by level, forming a ``wave'' of parallel insertions that sweeps from the top level down to the base level.

The merge phase consists of two main ingredients: (i)  A \emph{preprocessing phase}, executed on $\B$, and (ii) the \textsc{WAVE} merge protocol, executed on $\C$, in which the groups traverse and progressively integrate themselves into the skip list in parallel.

\paragraph*{Cohesive Groups} A \emph{cohesive group} $X =\{x_0,x_1,\dots,x_k\}\subseteq \B$ at some level $\ell$ is a set of contiguous nodes such that: (i) each node in $X$ has maximum height $\ell$, i.e., $\lmax{\B}{x_i} = \ell$, (ii) the group is sorted in increasing order.

Each group elects as a \emph{leader} the node with the smallest ID (i.e., value) in the group, which will coordinate the merge actions on behalf of the group. Led by the leader, the cohesive group traverses $\C$ just like a node being inserted into a skip list.  As the group moves through $\C$, it may split into subgroups because the cohesive group encountered a node $w$ in $\C$  with some elements of $X$ being smaller than $w$, while others are greater. Note that a split requires two rounds, one for the leader to inform all nodes in the new cohesive group who their leader is, and one for the nodes in the new cohesive group to ``introduce'' themselves to their new leader. We now observe a key structural property of how cohesive groups behave during traversal. 
\begin{proposition}[Cohesive Group Split Behavior]\label{proposition:m_cohev_behav}
	Let $X = \{x_0,x_1,\dots,x_k\}$ where $x_{i-1}\leq x_{i}$ for $i\in [1,k] $ be a set of nodes that are traversing the skip list $\C$ all together (i.e., as a cohesive group). And let $v\in X$ be a node that separates from $X$ at some time $t\geq 0$, then all $u\in X$ such that $u\geq v$ will separate from $X$ as well, and will go along with $v$.
\end{proposition}
\begin{proof}
	For the sake of contradiction, let $X = \{x_0,x_1,\dots ,x_k\}$ where $x_{i-1}\leq x_{i}$ for $i\in [1,k] $ be a cohesive group that is traversing the skip list. Moreover, assume that $x_j$ separates from $X$ at some time $t$, and that there exists $x_\ell$ for $\ell \geq j$ that does not follow $x_j$, i.e., $x_\ell$ does not separate from $X$. This generates a contradiction on the assumption that the elements in $X$ are sorted in ascending order.     
\end{proof}
Proposition~\ref{proposition:m_cohev_behav} implies that when a cohesive group splits during its traversal of \( \mathcal{C} \), the split always occurs at a single position along the ordered group. As a result, the group divides into exactly two contiguous subgroups: a left sub-group that continues the current traversal, and a right sub-group that separates and begins its own traversal (under a new leader). Therefore, at any round, a cohesive group can split into \emph{at most two} new cohesive groups.
Hence, we have that the traversal behavior of a cohesive group is tightly structured: in each round, the group either moves forward together, or splits cleanly into two contiguous subgroups. This ensures that group behavior remains predictable and parallelizable, with only one new leader elected per split. In particular, since group members perform the same operations up to the point of divergence, the cost of inserting a group is not significantly worse than inserting a single element. We formalize this next.
\begin{lemma}[Cohesive Group Insertion Time Bound]\label{lemma:choesive_group_insertion}
	Let $X = \{x_0,x_1,\dots,x_k\}$ where $x_{i-1}\leq x_{i}$ for $i\in [1,k] $ be a set of cohesive nodes that are inserted all together at the same time and from the top left sentinel node in a skip list $\C$. Then, the time to insert the entire group $X$ is at most twice the time needed to insert just $x_k$.
	
\end{lemma}
\begin{proof}
	We prove the lemma by induction on the number of nodes in the set $X$. As a base case we consider $X=\{x_0\}$. Since is the only node, its insertion follows its optimal search path from the top-left sentinel in $\C$. Therefore, the insertion time for $x_0$ is exactly the same wether inserted alone or as a part of a singleton group. For the inductive hypothesis, we assume that for some $k\geq 1$, the time to insert the group $X= \{x_0,x_1,\dots,x_{k-1}\}$ is at most twice the time needed to insert only $x_{k-1}$, and that each element in the group follows its optimal search path when inserted from the top-left sentinel.	For the inductive step we consider the insertion of $x_k$, along with the existing group $X=\{x_0,x_1,\dots,x_{k-1}\}$ (i.e., $X\cup\{x_k\}$). The key observation is that the search path of $x_k$ can be decomposed in two parts: (i) a shared prefix with $x_{k-1}$'s search path, and (ii) a unique suffix where $x_k$ diverges and continues alone towards its insertion point. Since the nodes move in parallel (cohesively) and are inserted from the same starting point, they make the same local decisions up to their divergence. By the inductive hypothesis, $x_{k-1}$'s search path is optimal, so the shared prefix must also be part of the optimal path for $x_k$. The remaining part (i.e., suffix) is also optimal by correctness of the skip list insertion algorithm~\cite{Pugh_1990}. Hence, the entire search path for $x_k$ is optimal. Since all insertions proceed in \emph{parallel}, the total insertion time is determined by the \emph{longest} search path among the group. By the inductive hypothesis, the insertion of $\{x_0,x_1,\dots,x_{k-1}\}$ takes at most two times the time to insert $x_{k-1}$ alone. Now we add $x_k>x_{k-1}$, whose path dominates $x_{k-1}$'s due to the group's ordering. Since both, shared and unique, parts of $x_k$'s path are optimal, the total time remains within twice the time to insert $x_k$ alone.
\end{proof}
This lemma ensures that even in the worst case, inserting a full group in parallel does not significantly increase the insertion time compared to a single node.
Next, we examine how they are positioned within the skip list, and how the skip list's structure ensures a clean separation between the levels. The following proposition captures two important structural properties of cohesive groups.
\begin{proposition}\label{proposition:cohesive_neighbors}
	Given a skip list $\B$, there is only one cohesive group at $\Lmax{\B}$. Moreover, for each $0\leq \ell <\Lmax{\B}$, let $X_\ell = \{x_0,x_1,\dots,x_k\}$ be a cohesive group at level $\ell$. Then, $X_\ell$'s left neighbor (i.e. $x_0$'s left neighbor at level $\ell$) is a node $v$ such that $\lmax{\B}{v}> \ell$.            
\end{proposition}
\begin{proof}
	Let $\ell$ be the maximum height of the skip list $\B$ and assume there are two disjoint cohesive groups $X_{\ell}$ and $Y_{\ell}$. This implies that $X_{\ell}$ and $Y_{\ell}$ are separated by one (or more) nodes $v$ such that $\lmax{\B}{v}>\ell$. Thus we have a contradiction on $\ell$ being $\B$'s maximum height. Next, let $0\leq \ell<\Lmax{\B}$, consider the cohesive group $X_{\ell}=\{x_0,x_1,\dots,x_k\}$ and its left neighbor $v$ at level $\ell$ (i.e., $v\in \neig{\B}{\texttt{left}}{\ell}{x_0}$). Then $v$ must be such that $\lmax{\B}{v}>\ell$. That is because, if $\lmax{\B}{v} = \ell$ then we would have $v\in X_\ell$ and if $\lmax{\B}{v} < \ell$, we would have $v\notin \neig{\B}{\texttt{left}}{\ell}{x_0}$.
\end{proof}
These properties are crucial, they imply that every cohesive group is ``encapsulated'' between nodes of higher height, and in particular, that each group has two natural \emph{parents} at the levels above. 

\paragraph*{Parent-Children Relationship in $\B$}

To coordinate the merge across levels, we define a parent-child relationship among nodes in $\B$. Given a node $u$ at level $\ell>0$, its \emph{children} at level $\ell'<\ell$ are defined as the maximal contiguous block of nodes adjacent to $u$ at level $\ell'$, whose heights are exactly $\ell'$, and such that all nodes between them (if any) have height at most $\ell'$. Formally, we define
\[
\childrenLevel{\ell'}{\texttt{left}}{u} =\{v : v<u \land \ell_{\max}(v) = \ell'\land \forall v<w<u , \ell_{\max}(w) \leq \ell' \}
\]
and
\[
 \childrenLevel{\ell'}{\texttt{right}}{u} =\{v : v>u \land \ell_{\max}(v) = \ell'\land \forall u<w<v , \ell_{\max}(w) \leq \ell' \}
\]
as the left and right children of $u$ at level $\ell'<\ell$ in $\B$, respectively. Moreover, we refer to $u$'s set of children at level $\ell'$ as the union of its left and right ones at level $\ell'$, i.e., $\ChildrenLevel{\ell'}{u} = \childrenLevel{\ell'}{\texttt{left}}{u} \cup \childrenLevel{\ell'}{\texttt{right}}{u}$ and to its overall set of children as the union of each $\ChildrenLevel{\ell'}{u}$ for all $\ell'<\ell$ , i.e., $\Children{u} =\bigcup_{\ell' = 0}^{\ell-1}\ChildrenLevel{\ell'}{u}$. Furthermore, we notice that two neighboring nodes $u$ and $w$ at level $\ell_{\max}^\B$ share a subset of their children.
Each node $u$ at level $\ell <\ell_{\max}^\B$ has \emph{two parents} at level $\ell'>\ell$ one to the left and one to the right. These parent nodes are responsible for notifying $u$ when it can safely initiate its own merge into $\C$ and when to update its pointers. 

This relationship forms the backbone of the merge dependency graph, ensuring that merge operations flow top-down through the skip list levels in a pipelined\footnote{By pipelined, we mean that they are executed one after the other.} manner. It also allows nodes that are not in the merge phase yet to track the progress of their parents and initiate their own actions only when dependencies are resolved.

Figure~\ref{fig:children_b} shows an example of the parent-children relationship in a skip list.

With this parent-children structure in place, the \textsc{WAVE} protocol guarantees that cohesive groups at lower levels only activate once their parents have successfully merged. This enables a clean pipelined traversal of $\C$ and maintains correctness and synchronization.

\begin{figure}[htb!]
	\captionsetup[subfigure]{justification=centering}
	\centering
	\begin{subfigure}{0.49\textwidth}
		\centering
		\includegraphics[scale=0.7]{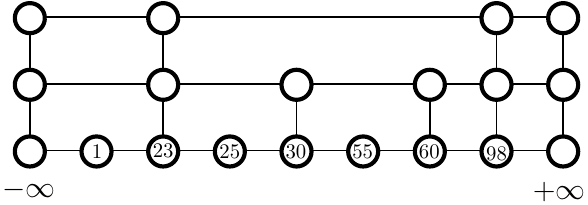}
		\caption{}\label{fig:children_a}
	\end{subfigure}
	\begin{subfigure}{0.49\textwidth}
		\centering
		\includegraphics[scale=0.7]{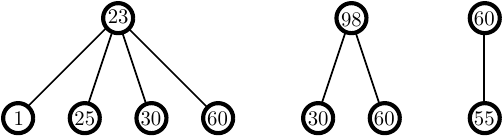}
		\caption{}\label{fig:children_b}
	\end{subfigure}
	\caption{Example of the parent-children relationship defined above. Figure~\ref{fig:children_a} is an example of a skip list. Figure~\ref{fig:children_b} shows the children set $\Gamma$ of nodes $23$, $98$, and $60$, respectively.}\label{fig:childrens}
\end{figure}

\paragraph*{The Preprocessing Phase in $\B$}

Before the \textsc{WAVE} protocol begins, we perform a \emph{preprocessing phase} in $\B$ to prepare each cohesive group for coordinated insertion. The phase includes:
\begin{description}
	\item[Group Identification.] At each level $\ell$, we identify all maximal cohesive groups. This can be done locally by the nodes as follows: let $X_\ell = \{x_0,x_1,x_2,\dots ,x_k\}$ be a set of consecutive nodes such that $\lmax{\B}{x} = \ell$ for each $x\in X_\ell$ and $0\leq \ell\leq \Lmax{\B}$. Then, each node $x_i\in X_\ell$ for $1\leq i\leq k$ sends its ID to $x_0$ through its left neighbor. Doing so will allow each $x_i$ to discover the ID of all the nodes $x_j\in X_\ell$ such that $x_j>x_i$ (see Figure~\ref{fig:preprocessing_buffer}).

	\begin{figure}[htb!]
		\captionsetup[subfigure]{justification=centering}
		\centering
		\begin{subfigure}{0.35\textwidth}
			\centering
			\includegraphics[scale=0.4]{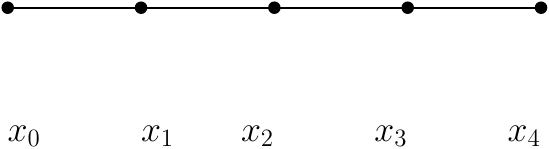}
			\caption{}\label{fig:preprocessing_a}
		\end{subfigure}
		\begin{subfigure}{0.3\textwidth}
			\centering
			\includegraphics[scale=0.4]{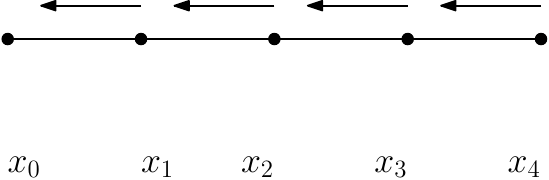}
			\caption{}\label{fig:preprocessing_c}
		\end{subfigure}
		\begin{subfigure}{0.3\textwidth}
			\centering
			\includegraphics[scale=0.4]{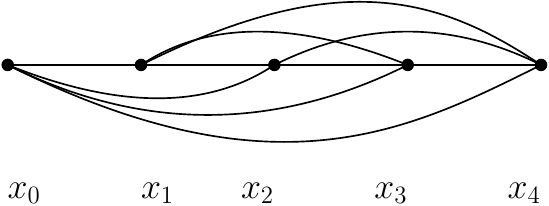}
			\caption{}\label{fig:preprocessing_b}
		\end{subfigure}
		\caption{Example of the preprocessing operation on a cohesive group of nodes. Figure~\ref{fig:preprocessing_a}, Figure~\ref{fig:preprocessing_c} and Figure~\ref{fig:preprocessing_b} show the cohesive group before, during, and after the preprocessing phase, respectively.}\label{fig:preprocessing_buffer}
	\end{figure}
	
	\item[Leader Election.] Within each group $X_\ell = \{x_0,\dots,x_k\}$, the leftmost node $x_0$ elects itself as the group leader, and broadcasts this information to the other group members. This operation takes $\Oo(1)$ rounds since distances among members in the cohesive groups have been shortened in the previous step.
	\item[Parent Discovery.] Each node identifies its two parents (left and right) at the level(s) above by examining their pointers in $\B$. These parent links define the merge dependency graph used in the \textsc{WAVE} protocol.
	\item[State Initialization.] All nodes initialize their local state as \emph{idle} except the leader on the \emph{single} cohesive group in $\B$'s top-level, which are immediately activated to begin the merge.
\end{description}
Each step above can be performed independently at every level of $\B$, and each node only participates in one group. Since the number communication is local within groups, the preprocessing phase needs $\Oo(\log n)$ rounds with high probability in the first step.
This follows from the fact that each cohesive group at some level $\ell$ contains at most $\Oo(\log n)$ nodes \whp~(see  Lemma~\ref{lemma:run_length} in Appendix~\ref{apx:skip_list}).
After this step each node within the same group will be at distance one from all the other nodes in the group. This implies that the remaining steps can be performed in constant number of rounds. This is a key property since it allows nodes to communicate with their children and the other nodes within their same group in $\Oo(1)$ rounds.

\begin{lemma}\label{lemma:preprocessing}
	Given a skip list of $n$ nodes, the preprocessing phase requires $\Oo(\log n)$ rounds \whp
\end{lemma}

With all cohesive groups identified, their leaders elected, and dependencies established through parent relationships, the system is now fully prepared to initiate the WAVE protocol for merging into the clean skip list.

\paragraph*{The \textsc{WAVE} Protocol}
Once the preprocessing completes, the merge proceeds via the \textsc{WAVE} protocol, a top-down, parallel traversal of the clean skip list $\C$ performed by cohesive groups in the buffer $\B$. The protocol resembles a wave-like propagation of merges: each level of $\B$ can become active only after their parents successfully started their merge phase. The \textsc{WAVE} protocol consist of three coordinated behaviors: active traversal, parallel insertion and virtual walking.

\begin{description}
	\item[Active Traversal.] At the beginning, only the unique top-level cohesive group in $\B$ is active. It initiates the merge by starting at the top-left sentinel of $\C$. The group leader coordinates a rightward traversal in $\C$ to spot the correct insertion point between two existing nodes in $\C$. 
	
	As the group moves, each step of the traversal is communicated downwards to its children in $\B$. Specifically, each node in the group broadcast its current position (e.g., the node it has just scanned in $\C$) to its child groups at the levels below. This communication is done in $\Oo(1)$ rounds because of the preprocessing phase and it provides the children with a ``virtual trace'' of the parent's path, allowing them to prepare for future activation as they were walking on $\C$ along with their parents.
	
	If the group encounters a node in $\C$ that triggers a split (e.g., due to key ordering), it divides into subgroups. Each new subgroup elects a new leader and continues the traversal independently.
	\item[Parallel Insertion.] Once a cohesive group reaches its insertion point at a given level in $\C$, the leader coordinates a local rewiring operation to insert the group between its predecessors and successors at that level in $\C$. This operation is executed in parallel by all the group members and takes $\Oo(1)$ rounds. After a successful merge at level $\ell$, the group leader notifies all its children in $\B$. These children are then potentially eligible to activate and begin their own traversal at the next level down.  
	\item[Virtual Walking.] Nodes in $\B$ that are not yet active remain in an \emph{idle} state. Each node continuously monitors its two parents at the levels above. As their parents perform their active traversal in $\C$, they broadcast their current search position downward, creating a \emph{trace} of the search path. 
	
	Once both parents complete their merge and communicate readiness, the idle node can change state to \emph{active} and begins its own merge process. 
	
	In addition, by observing these traces, an idle node can determine whether its own search path will likely diverge from those of its parents. If so, the node can safely activate early, using the most recent trace as a starting point. This ensures that the node's merge proceeds independently, without interfering with its parents’ ongoing insertion.
	
	This mechanism ensures that each level of $\B$ merges as soon as possible only after all dependencies have been satisfied, maintaining a clean, pipelined progression through the levels of the skip list.
\end{description}

The above description outlines the behavior of the \textsc{WAVE} protocol at a high level, highlighting its coordination mechanisms and activation rules. We now present a formal specification of the protocol, broken into two components: (1) Algorithm~\ref{algo:merge} that describes the behavior of active group leaders as they traverse and insert themselves into the clean skip list $\C$, and (2) Algorithm~\ref{algo:virtualwalk} that highlights the behavior of idle nodes during the virtual walking phase, where they observe parent traces and determine their activation point. Figure~\ref{fig:clean_buffer_before_merge} shows the Clean Network Buffer $\C$ before the merge phase (Figure~\ref{fig:wave_a}) and the Buffer Network $\B$ after the preprocessing phase (Figure~\ref{fig:wave_b}). While Figure~\ref{fig:wave_conceptual} provides a conceptual illustration of the \textsc{WAVE} protocol that merges $\B$ into $\C$. At each level, the advancing wave merges the corresponding nodes in $\B$ into $\C$, progressively building the integrated structure top-down. Furthermore, we refer to Appendix~\ref{apx:wave} for the visual representation of the step-by-step execution of the protocol on the networks in Figure~\ref{fig:clean_buffer_before_merge}.
\begin{figure}[htb!]
	\captionsetup[subfigure]{justification=centering}
	\centering
	\begin{subfigure}{0.4\textwidth}
		\centering
		\includegraphics[scale=0.7]{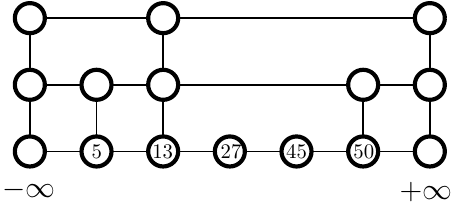}
		\caption{}\label{fig:wave_a}
	\end{subfigure}
	\begin{subfigure}{0.44\textwidth}
		\centering
		\includegraphics[scale=0.7]{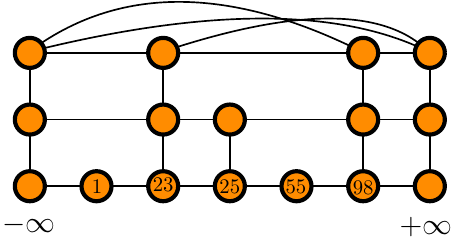}
		\caption{}\label{fig:wave_b}
	\end{subfigure}
	\caption{Figure~\ref{fig:wave_a} is the Clean Network $\C$ before the merge phase, while Figure~\ref{fig:wave_b} shows the Buffer Network $\B$ after the $\Oo(\log n)$ rounds preprocessing phase.}\label{fig:clean_buffer_before_merge}
\end{figure}

\begin{figure}[htb!]

		\centering
		\includegraphics[scale=0.7]{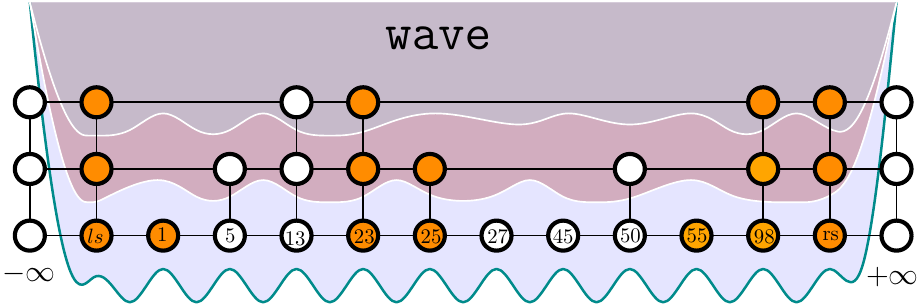}
	
	\caption{Conceptual representation of the {\sc WAVE} protocol. Each wave corresponds to a phase in which cohesive groups at a given level of $\B$ are merged into the clean skip list $\C$, proceeding top-down.}\label{fig:wave_conceptual}
\end{figure}

\subparagraph*{Detailed description of the \textsc{WAVE} Protocol.} In the first round of the merge protocol, the \emph{unique} cohesive group $X_{\Lmax{\B}}=\{x_0,x_1,\dots,x_k\}$ with $x_0=-\infty$ and $x_k = +\infty$ in $\B$'s topmost level changes its state to \texttt{merge} and starts its merging phase in $\C$ from the left-topmost sentry. For the sake of simplicity, we consider $\B$'s left and right sentries as nodes of the buffer network that must be merged with $\C$\footnote{This does not affect the merge outcome. Moreover, when $\B$ is completely merged with $\C$, we can remove $\B$'s sentinels from $\C$ in $\Oo(1)$ rounds.}. We assume that such sentinels are contained in $\C$'s ones. Moreover, at each round, every cohesive group $X\subseteq \B$ on a node $v\in\C$ at level $\ell$ that is in the \texttt{merge} performs a merge phase step that is composed by a Traversal Step in which the cohesive group moves on $\C$ and (if needed) a Merge Step in which $X$ merges in the current level $\ell$ in the clean network.

\begin{algorithm}[htb!]
	\caption{\textsc{Overview of the Merge Phase} for Group $X$}
	\label{algo:merge_loop}
	\KwIn{Cohesive group $X$ in \texttt{merge} state, leader $x_0$}
	\While{$X$ has not merged at level $0$}{
		Let $\ell \leftarrow$ current level of $X$\;
		\textsc{TraversalStep}($x_0$, $X$, $\ell$)\tcp*{$\Oo(1)$ Rounds.}
		\textsc{MergeStep}($X$, $\ell$)\tcp*{$\Oo(1)$ Rounds.}
		Update $X$'s current level to $\ell - 1$\;
	}
	Set state of all $u \in X$ to \texttt{done}\;
\end{algorithm}
Algorithm~\ref{algo:merge_loop} gives an overview of the step executed by a cohesive group in the merge phase. The group, until merging at level $0$, performs a traversal step followed by a potential merge step at the current level. Each of these steps requires $\Oo(1)$ number of rounds. Algorithm~\ref{algo:traversal} and Algorithm~\ref{algo:merge} provide a formal description of the steps performed by the nodes while merging. A key idea of the protocol is that every time a node in the \texttt{merge} state performs one move on $\C$, it communicates its progress to all its children. Notice, that because of the preprocessing phase parents can communicate with their children in constant number of rounds. Nodes, that are \texttt{idle} in $\B$ receiving updates from their parents can use this information to ``virtually walk'' on the clean network along their parents. Indeed, a node $v$, in the buffer network $\B$ in \texttt{idle} state is continuously listening to its parents. $v$'s behavior changes according to the type of message it receives. More precisely, let us assume that $v$ receives a message of the type $m(u)= \langle L,R,\texttt{move},\texttt{level} \rangle$ from its parent $u$: If $v$ is \texttt{idle} and discovers that both its parents successfully merged in $\C$ in level $\ell= \Lmax{v}$, then it must change state to \texttt{merge} and proceed with its own merge phase in $\C$ starting from $v.\texttt{starting\_point}$ at level $v.\texttt{level\_to\_start}$. Algorithm~\ref{algo:virtualwalk} describes the actions performed by the \texttt{idle} nodes in the buffer network upon receiving messages from their parents.

Before analyzing the merge runtime, we first observe that idle nodes only update their state using information from their parents. This ensures the correctness of their updated search path:

\begin{lemma}
	Let $v$ be an idle node in the buffer skip list $\mathcal{B}$ at level $\ell$. Throughout the execution of the \textsc{WAVE} protocol, every update to $v.\texttt{starting\_point}$ results in a node that lies on the optimal search path for $v$ in the augmented skip list $\mathcal{C}' = \mathcal{C} \cup \mathcal{B}_{>\ell}$. Where $\C' = \C\cup \B_{>\ell}$ is obtained by merging $\C$ and the elements of $\B$ situated in the level above $\ell$. 
\end{lemma}

\begin{proof}
	Each idle node $v$ receives traversal updates exclusively from its two parents in $\B$, which reside at levels strictly greater than $\ell$. These messages reflect nodes visited during the parents' optimal search paths while merging into $\C$. Since $v$'s own optimal search path in $\C'$ would also pass through the same upper-level nodes before descending, any such updated $\texttt{starting\_point}$ lies on a valid prefix of $v$'s optimal path. Thus, the virtual walking behavior maintains correctness.
\end{proof}
\begin{algorithm}[htb!]
	\caption{\textsc{TraversalStep} for Group Leader $x_0$}
	\label{algo:traversal}
	\KwIn{Cohesive group $X$ at level $\ell$ in $\mathcal{B}$, leader $x_0$}
	\KwOut{Traversal continuation or split + merge at level $\ell$}
	
	$v \leftarrow$ current position of $x_0$ in $\mathcal{C}$\;
	$z \leftarrow$ right neighbor of $v$ at level $\ell$ in $\mathcal{C}$\;
	Broadcast $z$ to all followers in $X$\;
	
	$X_{\texttt{right}} \leftarrow \{ u \in X \mid u > z \}$\;
	$X_{\texttt{left}} \leftarrow X \setminus X_{\texttt{right}}$\;
	
	\If{$X_{\texttt{right}} \neq \emptyset$}{
		Elect new leader $x'_0 \leftarrow \min(X_{\texttt{right}})$\;
		\ForEach{$u \in X_{\texttt{right}}$}{
			Send message $\langle v, z, \text{``right''}, \ell \rangle$ to $\Gamma(u)$\tcp*{$\Oo(1)$ rounds.}
		}
		\textsc{TraversalStep}($x'_0$, $X_{\texttt{right}}$, $\ell$)\;
	}
	
	\If{$X_{\texttt{left}} \neq \emptyset$}{
		\ForEach{$u \in X_{\texttt{left}}$}{
			Send message $\langle v, z, \text{``down''}, \ell \rangle$ to $\Gamma(u)$      \tcp*{$\Oo(1)$ rounds.}

		}
		\textsc{MergeStep}($X_{\texttt{left}}, \ell$)\;
	}
\end{algorithm}

\begin{algorithm}[htb!]
	\caption{\textsc{MergeStep} for Group $X$}
	\label{algo:merge}
	\KwIn{Group $X$ at level $\ell$ in $\mathcal{B}$}
	\KwOut{Group merged at level $\ell$ and notified children}
	
	Insert $X$ between $v$ and $z$ in $\mathcal{C}$ at level $\ell$\;
	\ForEach{$u \in X$}{
		$y \leftarrow$ right neighbor of $u$ in newly merged segment\;
		Send message $\langle u, y, \text{``down''}, \ell \rangle$ to $\Gamma(u)$\tcp*{$\Oo(1)$ rounds.}
		Notify $\Gamma(u)$ of successful merge\tcp*{$\Oo(1)$ rounds.}
	}
	
	\uIf{$\ell > 0$}{
		\If{$X$ has neighbors in $\mathcal{B}$ at level $\ell$}{
			Wait for neighbors to be merged at level $\ell - 1$\;
		}
		\textsc{TraversalStep}(\texttt{leader}($X$), $X$, $\ell - 1$)\;
	}
	\Else{
		Set state of all $u \in X$ to \texttt{done}\;
	}
\end{algorithm}

\begin{algorithm}[htb!]
	\caption{\textsc{VirtualWalk} for Idle Node $v$}
	\label{algo:virtualwalk}
	\KwIn{Messages $m(u)$ from parent(s)}
	\KwOut{Node state update and possible activation}
	
	\ForEach{received $m(u) = \langle L, R, \texttt{move}, \texttt{level} \rangle$}{
		Update $v.\texttt{starting\_point}$ and $v.\texttt{level\_to\_start}$ based on $L$, $R$\;
		Update parent locations depending on move\;
	}
	
	\If{$v$ is independent from both parents}{
		\uIf{left neighbor $z$ shares same starting point}{
			Set $v$ as follower\;
		}
		\Else{
			Set $v$ as leader\;
			Form group $X$ with $v$ and right neighbors sharing starting point\;
			\textsc{TraversalStep}($v$, $X$, $v.\texttt{level\_to\_start}$)\;
		}
	}
	
	\If{both parents merged at level $\ell = \Lmax{v}$}{
		Set $v$ to \texttt{merge}\;
		Proceed with group formation and merging as above\;
	}
\end{algorithm}

\paragraph*{Analysis of the \textsc{WAVE} Protocol}

A key aspect of the protocol is that a cohesive group $X$, after being successfully merged at level $\ell$, may need to wait before proceeding to level $\ell-1$. This delay occurs if one of $X$'s neighbors in $\C$ at level $\ell$ is a node from the buffer $\B$ (i.e., one of its parents) that has not yet completed its merge at level $\ell-1$. The following lemma provides an upper bound on the time required for such a group $X$ to fully merge into $\C$, accounting for possible dependencies on its parents at each level.

\begin{lemma}\label{lemma:merge_generalized}
	Let $X$ be a cohesive group at level $0\leq \ell\leq \Lmax{\B}$ in $\B$. The merge time of $X$ in $\C$ takes $\Oo(\log n+\Lmax{\B}-\ell)$ rounds \whp
\end{lemma}
\begin{proof}
	Assume that the cohesive group $X$ was merged at level $\ell$ of the clean network and has to proceed with its own merge at level $\ell-1$. Let $v$ and $z$ be $X$'s left and right parents in $\B$. We have three cases to consider: (i) $X$ is independent from its parents; (ii) $X$ depends on at least one of its parents through all the levels $\ell'<\ell$; and, (iii) $X$ may depend on one of its parents for some time and become independent at some point. In the first case, $X$ is merging itself in a sub-skip list that contains only elements in the clean network. Thus, $X$'s overall merging time requires $\Oo(\log n)$ \whp~To analyze the second case, we notice that as $X$ depends on its parents, also $X$'s parents may depend on their parents as well and so on. We can model the time a cohesive group $X$ at level $\ell$ has to wait in order to start its own merge procedure at level $\ell-1$ as the time it would have taken had it traveled in a pipelined fashion with its parents. Moreover, such pipelining effect can be extended to include, in a recursive fashion, the parents of $X$'s parents, the parents of the parents of $X$'s parents, and so on. Thus, $X$ can be modeled as the tail of such ``parents-of-parents'' chain (or pipeline), in which $X$ must wait for all the nodes in the pipeline to be merged at level $\ell-1$ before proceeding with its own merge in such level. Moreover, we observe that, by our assumption, $X$ is dependent on at least one of its parents throughout the merge phase. Without loss of generality, let us assume that $X$ depends on one parent, say $z$. In other words, $X$ is never splitting, and its placing itself right after/before $z$ at each level $\ell'<\ell$. The time $X$ takes to merge in one level $\ell'$ after $z$ is $\Oo(1)$ (see Lemma~\ref{lemma:choesive_group_insertion}). Thus, $X$ merge itself in some level $\ell'$ by time $T_z^{\ell'}+\Oo(1)$ where $T_z^{\ell'}$ is the time $z$ needs to merge itself in level $\ell'$. Furthermore, $z$ may be dependent to one of its parents, say $w$. Thus the time to insert $z$ at level $\ell'$ is $T_z^{\ell'} = T_w^{\ell'} +\Oo(1)$. If we repeat this argument for all the levels $\hat{\ell}>\ell'$ we obtain that the time to merge the cohesive group $X$ at level $\ell'$ can be defined as $T_y^{\ell'} + \Lmax{\B} - \ell+1 = T_y^{\ell'} + \Lmax{\B} - \lmax{\B}{X}+1$ where $y$ is the topmost parent in the pipeline of dependencies. Thus $X$'s overall merge time is $T_X^{0} = T_y^{0}+  \Lmax{\B} - \lmax{\B}{X}+1 = \Oo(\log n + \Lmax{\B} - \lmax{\B}{X}+1) $ \whp~In the last case, $X$ may depend on its parents for some time and then gain independency. Assume that $X$ depends its parent $z$ until level $\ell'$ and then gains independency. $X$'s overall merge time is $T_X^{0} = \Oo(T_z^{\ell'} + 1 + \log n) = \Oo(\log n)$ \whp
\end{proof}
Given Lemma~\ref{lemma:merge_generalized}, we can bound the overall running time of the {\sc WAVE} protocol.
\begin{lemma}\label{lemma:merge_wave}
	Let $\C$ and $\B$ be two skip lists of $n$ elements built using the same $p$-biased coin. Then,
	the WAVE protocol merges $\C$ and $\B$ in $\Oo(\log n)$ rounds \whp.
\end{lemma}
\begin{proof}
	The proof follows by noticing that a \texttt{idle} cohesive group $X$ in some level $\ell$ starts its merge procedure when the wave sweeps through it. This happens within $\Oo(\log n)$ rounds \whp~from the {\sc WAVE} protocol initialization time instant. Moreover, a \texttt{idle} node/cohesive group, after entering the \texttt{merge} state needs $\Oo(\log n+\Lmax{\B}-\ell)$ rounds \whp~to be fully merged in $\C$ (Lemma~\ref{lemma:merge_generalized}). That is, the merge time is at most $\Oo(\log n +\log n+\Lmax{\B}-\ell) = \Oo(\log n)$ \whp
\end{proof}

In summary, the \textsc{WAVE} protocol enables a fully distributed and parallel merge of the buffer skip list $\B$ into the clean skip list $\C$, using a top-down wave of coordinated insertions. Cohesive groups ensure local structure and predictability, while the parent-children dependencies across levels enable global synchronization. Through virtual walking and early activation, the protocol supports a high degree of parallelism and completes in $\Oo(\log n)$ rounds with high probability.

%% file: trunk/update.tex
\subsection{Update}~\label{sec:duplicate}
After performing the previous phases, the live network $\Ll$ must be coupled with the newly updated clean network $\C$. A high level description on how to perform such update in constant time is the following: during the duplication phase a ``snapshot'' (e.g., a copy) of the clean network $\C$ is taken and used as the new $\Ll$. Intuitively, this procedure requires $\Oo(1)$ rounds because each node in $\C$ is taking care of the snapshot that can be considered as local computation. However, we must be careful with the notion of snapshot. To preserve dynamic resource-competitiveness we show how to avoid creating new edges while taking a snapshot of $\C$ during this phase. Assume that every edge in the clean network $\C$ has a length two local\footnote{Meaning that each node $u\in \C$ has its own copy of the label $\lambda$ for each edge incident to it. Observe that for an edge $(u,v)$ the local label of $u$ can not differ from the one of $v$.} binary label $\lambda$ in which the first coordinate indicates its presence in $\C$ and the second one in $\Ll$. Moreover, each label $\lambda$ can be of three different types: (1) ``$10$'' the edge is in $\C$  not in $\Ll$; (2) ``$01$'' the edge is not in $\C$ but it is in $\Ll$; and, (3) ``$11$'' the edge is in both $\C$ and $\Ll$. Moreover, assume that each node $u$ in $\C$ has the label $\lambda(u,v)$ associated to its port encoding the connection with node $v$. Then, each node $u$ in $\C$ for each $0\leq \ell\leq \Lmax{u}$ performs the following local computation: 
\begin{description}
	\item[1.] If $u$'s port-label associated to the edge connected to the left neighbor at level $\ell$ (i.e., $\lambda(u,v)$ such that $v\in \neig{\C}{\texttt{left}}{\ell}{u}$) is ``$10$'' then set it to ``$11$'';
	\item[2.]If $u$'s port-label associated to the edge connected to the right neighbor at level $\ell$ (i.e., $\lambda(u,w)$ such that $w\in \neig{\C}{\texttt{right}}{\ell}{u}$) is ``$10$'' then set it to ``$11$'';
\end{description}
Notice that during this phase no edge has the label ``$01$'' because $\C$ was ``cleaned'' during the deletion phase. All the nodes in $\C$ executed the labeling procedure, and the live network $\Ll$ is given by the edges with label ``$11$''. 
It follows that this phase requires a constant number of rounds.
\begin{lemma}\label{lemma:duplicate}
	The update phase requires constant number of rounds.
\end{lemma}
Next, we show that all the phases in the maintenance cycle satisfy the dynamic resource competitiveness constraint defined in Section~\ref{sec:contributions}.
\begin{lemma}\label{lemma:workload}
    Each maintenance cycle is $(\alpha,\beta)$-dynamic resource competitive with $\alpha = 
    \Oo(\log n)$ and $\beta = \Oo(\Polylog{n})$.
\end{lemma}
\begin{proof}
	In order to prove that our maintenance protocol is dynamic resource competitive it suffices to show that a generic iteration of our maintenance cycle respects this invariant. 
	    
     Let $t_s$ be the time instant in which the deletion phase starts, and let $t_e= \Oo(t_s+\log n)$. The overall churn between the previous deletion phase and the current one is $C(t_s-\Oo(\log n),t_s) =\Oo(n)$, that is because we have $\Oo(n/\log n)$ churn at each round for $\Oo(\log n)$ rounds. The total amount messages sent during the phase is  $\tilde{\Oo}(C(t_s-\Oo(\log n),t_s))$ and the number of formed edges is  $\tilde{\Oo}(C(t_s-\Oo(\log n),t_s))$. Thus the overall amount of work during the deletion phase is $\W(t_s,t_e) =  \tilde{\Oo}(C(t_s-\Oo(\log n),t_s))$. 
	
	During the buffer creation phase, we create a twin-butterfly network $\M$ in $\Oo(\log n)$ rounds using $\Oo(\log n)$ messages for each node at every round. Moreover, the twinbutterfly has $\Oo(n\log n)$ nodes of which $2n$ nodes have degree $4$ while the rest have degree $8$. This implies that while building the twinbutterfly we create roughly $\Oo(n\log n)$ edges and we exchange $\Oo(n\log^2 n)$ messages. Next, during each round of the AKS algorithm, each node $u$ in $\M$ sends and receives a constant number of messages. Then, we copy such a list $\Oo(\log n)$ times \whp, in this way we create $\Oo(n\log n)$ edges. Finally, in the last step of the buffer creation phase, we run the deletion algorithm used in the previous phase on a subset of the nodes in the sorted list. Putting all together, during this phase, we have an overall amount of work of $\W(t_s,t_e) = \Oo(C(t_s-\Oo(\log n),t_s)\cdot\Polylog{n}) = \tilde{\Oo}(C(t_s-\Oo(\log n),t_s))$.
	
	During the merge phase, we merge the buffer network with the clean one. The preprocessing step creates $\Oo(n\log n)$ edges on the buffer network using at most $\Polylog{n}$ number of messages. During each step of the WAVE protocol sends $\Oo(n\cdot \Polylog{n})$ overall number of messages on the buffer network and creates $\Oo(n\log n)$ edges in the clean network. Putting all together, we have that the merge phase does not violate our dynamic work efficiency requirements. Indeed the overall work is bounded by 
	$\W(t_s,t_e) = \Oo(C(t_s-\Oo(\log n),t_s)\cdot\Polylog{n}) = \tilde{\Oo}(C(t_s-\Oo(\log n),t_s))$.
	
	Finally, during the duplication phase there is no exchange of messages and no edge is created. 
	
	We have that each phase of the maintenance cycle does not violate our resource competitiveness  requirements in Definition~\ref{def:dynamic_resource_comp}.
\end{proof}
We are now ready to show that the skip list is maintained with high probability for $n^d$ rounds, where $d\geq 1$ is an arbitrary big constant.
\begin{lemma}\label{lemma:resilience_whp}
The maintenance protocol ensures that the resilient skip list is maintained effectively for at least $\Poly{n}$ rounds with high probability. 
\end{lemma}
\begin{proof}
	Each phase of the maintenance protocol (i.e., Algorithm~\ref{algo:overview}) succeeds with probability at least $1-\frac{1}{n^d}$, for some arbitrary big constant $h\geq 1$. Let $X$ be a geometric random variable of parameter $p =\frac{1}{n^d}$ that counts the number of cycle needed to have the first failure, then its expected value is $\expect{X}  = n^d$, for $d\geq 1$. Moreover, the probability of maintenance protocol succeeding in a round $n^r$ for some $r<d$ is $(1-\frac{1}{n^d})^{n^r}$, Without loss of generality assume $d>1$ and $r = d-1$, then $(1-\frac{1}{n^d})^{n^{d-1}} \leq e^{-n^{d-1}/n^d}\geq 1-\frac{1}{n}$. 
\end{proof}
We conclude by proving the main theorem in Section~\ref{thm:overall_result}.
\subparagraph*{Proof of the Main Theorem (Theorem~\ref{thm:overall_result}).}
To prove the main theorem it is sufficient to notice that the skip list will always be connected and will never lose its structure thanks to Lemma~\ref{lemma:taking_over}. Indeed, every time a group of nodes is removed from the network, there are committees that in $\Oo(1)$ rounds will take over and act on their behalf. This allows all the queries to go through without any slowdown. Thus all the queries will be executed in $\Oo(\log n)$ rounds with high probability even with a churn rate of $\Oo(n/\log n)$ per round. Next we can show that the distributed data structure can be efficiently maintained using maintenance cycles of $\Oo(\log n)$ rounds each. To do this, it is sufficient to show that each phase of our maintenance cycle (Algorithm~\ref{algo:overview}) can be carried out in $\Oo(\log n)$ rounds (see Lemma~\ref{lemma:delete}, Lemma~\ref{lemma:creation}, Lemma~\ref{lemma:merge_wave}, and Lemma~\ref{lemma:duplicate}). Moreover, from Lemma~\ref{lemma:workload} we have that each phase is $(\alpha,\beta)$-dynamic resource competitive with $\alpha = \Oo(\log n)$ and $\beta = \Oo(\Polylog{n})$. Finally, the maintenance protocol ensures that the distributed data structure is maintained for at least $n^c$ rounds with high probability where $c\geq 1$ is an arbitrarily large constant.


%% file: trunk/generalization.tex
\subsection{Extending our approach to other data structures}\label{sec:generalize_ds}
The outlined maintenance protocol provides a convenient framework for constructing churn-resilient data structures. Furthermore, it can be easily adapted to maintain a more complex distributed pointer based data structure such as a skip graphs~\cite{Harvey_2003,Aspnes_2007,Goodrich_2006,Jacob_2014}. Indeed, with minor adjustments to the phases outlined in Algorithm~\ref{algo:overview}, a maintenance cycle capable of maintaing such data structures against \emph{heavy churn rate} can be devised. To this end, we briefly discuss how to adapt our results to skip graphs. A skip graph~\cite{Aspnes_2007}, can be viewed as an extension of skip lists. Indeed, both consists of a set of increasingly sparse doubly-linked lists ordered by levels starting at level $0$, where membership of a particular node $u$ in a list at level $\ell$ is determined by the first $\ell$ bits of an infinite sequence of random bits associated with $u$, referred to as the \emph{membership vector} of $u$, and denoted by $m(u)$. Let the first $\ell$ bits of $m(u)$ as $m(u)|\ell$. In the case of skip lists, level $\ell$ has only one list, for each $\ell$, which contains all elements $u$ such that $m(u)|\ell=1^\ell$, i.e., all elements whose first $\ell$ coin flips all came up heads. Skip graphs, instead, have $2^\ell$ lists at level $\ell$, which we can index from $0$ to $2^\ell-1$. Node $u$ belongs to the $j$-th list of level $\ell$ if and only if $m(u)|\ell$ corresponds to the binary representation of $j$. Hence, each node is present in one list of every level until it eventually becomes the only member of a singleton list. Without loss of generality, assume that the skip graph has sentry nodes as classical skip lists\footnote{Thus, at each level $\ell$ there are $2^\ell$ sentinels (one for each list in such level).}. Thus, the skip graph can be maintained as follows:
\begin{description}
    \item[Deletion.] We can use the same $\Oo(\log n)$ rounds approach described in Section~\ref{sec:delete} to remove from the \emph{clean skip graph} the nodes that have left the network.
    \item[Buffer Creation.] To build the base level of the \emph{buffer skip graph} we can use the same $\Oo(\log n)$ rounds approach described in Section~\ref{sec:creation}. While, to construct the skip graph from level $0$, we build an ``augmented'' skip list in which each node $u$ in each level $\ell$ is identified by its ID/key and the list $j$ at level $\ell$ in which $u$ appears (there are $2^\ell$ lists at level $\ell$). Thus, each node $u$ at some level $\ell$ can be: (i) \emph{effective} in the list $j$ and \emph{fill-in} for the remaining $i\neq j$, for $j,i\in  [1,2^\ell]$; or, (ii) \emph{fill-in} in all the lists $j\in [1,2^\ell]$. Finally, such augmented skip list is transformed into a skip graph by running (in parallel) Algorithm~\ref{algo:delete} in Section~\ref{sec:delete} on each level $\ell$ for which there exists at least one fill-in node.
    \item[Merge.] We can use the same merge phase described in Section~\ref{sec:merge}, with the only difference that a node $u$ in the \texttt{merge} state at some level $\ell$, must be merged in one of the $2^\ell$ list in such level.
    \item[Duplicate.] We can use the same approach described in Section~\ref{sec:duplicate}.
\end{description}

Consequently, we can claim a similar result to Theorem~\ref{thm:overall_result} for skip graphs. In essence, with small changes in the delete, buffer creation, merge and update algorithms, our maintenance algorithm can be tailored to the specific data structure in question. Thus, making our approach ideal for building complex distributed data structures in highly dynamic networks.
\begin{corollary}
\label{cor:skip_graphs}
The maintenance cycle described in Algorithm~\ref{algo:overview} can be adapted to support skip graphs. In particular, the resulting data structure remains resilient to churn at a rate of up to $ \Oo(n / \log n) $ per rounds, and all operations complete within $ \Oo(\log n) $ rounds with high probability, using at most $ \Oo(\Polylog{n}) $ messages per node per round.
\end{corollary}

%% file: trunk/general_keys.tex
\subsection{Extending our approach to deal with multiple keys on each node}\label{sec:generalize_keys}
Throughout the paper we assumed that each node in the network posses \emph{exactly} one element of the data structure. This assumption can be relaxed to deal with multiple keys on each node. Indeed, in the case in which in the network there are $t\cdot n$ keys, where $t=\Polylog{n}$, all the techniques described above allow to maintain the data structure in the presence of the same adversarial churn rate. 
\begin{corollary}\label{cor:generalization}
    Given a network with $n$ nodes in which each vertex $v$ possesses $t= \Polylog{n}$ elements in the skip list. Then maintenance protocol requires $\Oo(\log n)$ rounds to build and maintain a resilient skip list that can withstand heavy adversarial churn at a churn rate of up to $\Oo(n/\log n)$ nodes joining/leaving per round. 
\end{corollary} 
The corollary follows from the fact that the adversary can assign the $t$ elements to each node in the network in such a way that each of their elements will be in a different cohesive group. Indeed, the first three phases of the maintenance cycle (i.e., Delete, Creation and Merge) will require $\Oo(\log (t\cdot n))$ rounds and for $t = \Polylog{n}$ we have that the running time of the maintenance cycle is $\Oo(\log (\Polylog{n}\cdot n)) = \Oo(\log \log^k n+\log n)=\Oo(\log n)$, where $k>0$.
Furthermore, to deal with larger $t$ the protocol must be slightly adjusted to take care of the higher amount of keys per node. Here, we do not delve into this problem since it boils down to an implementation detail that can be taken care of while implementing the protocol on real world P2P Networks. 

%% file: trunk/implications.tex
\section{Implications and Generalization of the Maintenance Framework}
We showed, for the first time, that complex pointer-based data structures can be maintained efficiently in the DNC model under near-linear adversarial churn. Our four-phase maintenance cycle provides the foundation for a general approach to maintaining overlays in highly dynamic environments. In Section~\ref{sec:applications} we propose some possible adaptations of our results to other distributed computing problems in the DNC model. In Section~\ref{sec:general_fw} we rephrase our results in terms of general framework for any distributed data structure that satisfies some specific properties. Finally, we propose a classification of distributed data structures in the DNC model that depends on the churn rate that they can tolerate.

\subsection{Applications of the Maintenance Framework}\label{sec:applications}
We describe several case studies illustrating how our maintenance protocol could allow other distributed data structures to be maintained in highly dynamic networks.
\paragraph*{\textbf{Maintaining expander graphs.}}
As an example of graph maintenance we briefly discuss how to build and maintain a (constant-degree) expander graph from an \emph{arbitrary} connected graph. We will outline how this can be accomplished using an $\Oo(\log n)$ rounds maintenance cycle, which is a consequence of this work and prior works. First, during the bootstrap phase, we build our overlay maintenance network in which each node $u$ has access to a set of well-mixed node IDs\footnote{Sampling from a set of well-mixed tokens is equivalent to sampling uniformly at random from the set $\{1,\dots, n\}$.}. Once such a network has been built, the bootstrap phase continues with the constant degree expander construction. Such a task can be accomplished using (for example) the \emph{Request a link, then Accept if Enough Space} (RAES) protocol by Becchetti et al.,~\cite{Becchetti_2020} with parameters $d\geq 1$ and $c\geq 2$. This technique builds a constant degree expander (in which all the nodes have degrees between $d$ and $c\cdot d$) in $\Oo(\log n)$ rounds and using overall $\Oo(n)$ messages \whp~Once the constant degree expander is constructed, the bootstrap phase ends and the adversary begins to exert its destructive power on the overlay network. Using our data structure maintenance protocol in parallel with the expander maintenance by Augustine et al.,~\cite{Augustine_2015_b} we can maintain a constant degree expander using $\Oo(\log n)$ rounds maintenance cycles in which each node sends and receives $\Oo(\Polylog{n})$ messages at each round.
\paragraph*{\textbf{Minimum Spanning Tree Maintenance.}} The minimum spanning tree (MST) problem can be solved  efficiently in the DNC model. In the MST problem we are given an arbitrary connected undirected graph $G$ with edge weights, and the goal is to find the MST of $G$. This can be accomplished 
using a $\Polylog{n}$ rounds bootstrap phase and maintenance cycles. During the bootstrap phase, we build the overlay churn resilient network in which each node has access to a set of well-mixed tokens and we build a constant-degree expander $H$ overlay on the given graph $G$ (the expander edges are added to $G$'s edge set). For this, we convert the expander (that is not addressable) into a butterfly network (that is addressable) which allows for efficient routing between any two nodes in $\Oo(\log n)$ rounds. This conversion can be accomplished using techniques of~\cite{Angluin_2005,Gmyr_2017,Gotte_2019,Gotte_2021,Augustine_2021}. All these protocols takes $\Polylog{n}$ rounds and $\tilde{\Oo}(n)$ messages to convert a constant-degree expander into an hypercubic (i.e., butterfly) network. Using the addressable butterfly on top of $G$, we can efficiently implement the Gallagher-Humblet-Spira (GHS) algorithm~\cite{Gallager_1983} as shown by Chatterjee et al.~\cite{Chatterjee_2020} to compute the MST of $G$ in $\Polylog{n}$ rounds and $\tilde{\Oo}(n)$ messages using routing algorithms for hypercubic networks~\cite{Valiant_1982,Upfal_1992}. After the bootstrap phase, we maintain (and update) the MST using our data structure maintenance protocol, the expander maintenance technique described above, and the MST computation techniques used in the bootstrap phase.

\paragraph*{\textbf{Maintaining Skip Graphs.}} Another implication is the construction and maintenance of other more sophisticated data structures like skip graphs~\cite{Aspnes_2007,Goodrich_2006,Jacob_2014}. All these data structures are not resilient to churns, and their maintenance protocols are not fast enough to recover the data structure after the failure of some nodes. Indeed, these protocols (see~\cite{Aspnes_2007,Goodrich_2006,Jacob_2014}) might need $\Oo(n)$ rounds to rebuild the skip-graph and they strictly require no additional churn to be able to fix the data structure. Our maintenance protocol overcomes these problems and provides a $\Oo(\log n)$ rounds skip graph constructing protocol and a technique able to repair them in $\Oo(\log n)$ in the presence of an almost linear churn of $\Oo(n/\log n)$ \emph{at every round}. Moreover, our maintenance mechanism allows for the users to query the data structure while being maintained.

\paragraph*{Beyond Distributed Computing.} We believe that some ideas from our results could be used in the centralized batch parallel setting to quickly insert batches of new elements in skip lists (or skip graphs) data structures (see e.g.,\cite{Tseng_2019}) and in the fully dynamic graph algorithms settings (see for example the survey~\cite{Hanauer_2022}) to perform fast updates of fully dynamic data structures. Indeed, in principle (provided the right amount of parallelism), our deletion and merge algorithms could be implemented in a parallel (centralized) setting and used to speed-up all kinds of computations involving these specific data structures.

\subsection{A General Framework for Churn-Resilient Structures}\label{sec:general_fw}
On a broader picture, the above applications can be unified under a single general framework for maintaining pointer-based data structures under adversarial churn. Given a data structure $\mathcal{D}$, we can ``abstract away'' the ideas in Algorithm~\ref{algo:overview} and obtain the following \emph{abstract maintenance cycle} for a generic data structure $\mathcal{D}$:
\begin{enumerate}
	\item \textbf{Delete Phase:} \label{step:delete}
	Identify and remove corrupted, outdated, or disconnected regions of $\mathcal{D}$. Deletion ensures that inconsistencies caused by churn do not propagate through the structure.
	
	\item \textbf{Creation Phase:} \label{step:create}
	Organize the set of newly arrived nodes into a provisional structure $\mathcal{B}$ (the buffer data structure). The goal is to prepare these nodes for integration, typically by arranging them into a sorted or partially structured form.
	
	\item \textbf{Merge Phase:} \label{step:merge}
	Integrate the buffer network $\mathcal{B}$ into the main structure $\mathcal{D}$ using a distributed merge protocol. This phase reconstructs the structure while respecting existing invariants.
	
	\item \textbf{Update Phase:} \label{step:update}
	Perform any necessary local corrections, including pointer rebalancing, level adjustments, or redundancy restoration, to finalize the integration.
\end{enumerate}
Provided that we have $\Oo(T)$-round distributed algorithms for each phase of the above abstract maintenance cycle, we can maintain $\mathcal{D}$ against an adversarial churn rate of $\Oo(n/T)$ per round. 
\begin{theorem}
	Let $\mathcal{D}$ be a distributed pointer-based data structure maintained using our four-phase cycle, and let $T$ be the maximum number of rounds needed for any phase. Then $\mathcal{D}$ can tolerate an adversarial churn rate of up to $\Oo(n/T)$ nodes per round, while preserving global correctness for at least $n^c$ rounds \whp, for any fixed constant $c > 0$.
\end{theorem}
The proof of the theorem follows the same steps as the proof of Theorem~\ref{thm:overall_result} by considering the $\Oo(T)$-round algorithms for the deletion, buffer creation, merge, and update phases for the specific data structure $\mathcal{D}$. 

\subparagraph*{Classes of Churn-Resilient Data Structures.} Our abstract maintenance cycle allows to define classes of churn-resilient data structures in the DNC model, where each class is characterized by the churn rate per round (thus the maintenance cycle runtime) that the data structure can tolerate while still supporting efficient update and query operations. Formally, let $t$ be the maintenance cycle run time function. We say that a distributed data structure $\mathcal{D}$ is \textbf{$t$-maintainable} if it can be successfully maintained by a $t$-round maintenance cycle. 
\begin{observation}
 Under the above general framework, every $t$-maintainable distributed data structure tolerates an adversarial churn rate of $\Oo(n/t)$ per round.
\end{observation}
In this paper, we showed that Skip-Lists (Section~\ref{sec:architecture}) and Skip-Graphs (Section~\ref{sec:generalize_ds}) are both $\Oo(\log n)$-maintainable and that are robust against and oblivious adversarial churn rate of $\Oo(n/\log n)$ per round. We conjecture that our bounds are tight, in the sense that $\log n$ is also a lower bound for the maintenance cycles for these data structures. That is because, in order to beat the $\Omega(\log n)$ barrier we would need to solve distributed sorting faster than in $\log n$ rounds while maintaining $(\alpha,\beta)$-dynamic resource competitiveness. 

This classification also raises a number of intriguing open problems. For example, it remains unclear whether there exist distributed data structures in the DNC model that are $\log\log n$-maintainable, and more generally, how to characterize the precise boundaries between maintainability classes. A key challenge is to establish lower bounds on the maintenance cycle time for fundamental distributed data structures. Do entirely new data structures need to be designed to exploit faster maintenance cycles?  We believe that our formulation of maintainability classes in the DNC model opens up a rich landscape for further exploration.

%% file: trunk/related_works.tex
\section{Related Works}
There has been a significant prior work in designing peer-to-peer (P2P) networks that 
can be efficiently maintained (e.g. see ~\cite{Aspnes_2007,Awerbuch_2004,Bhargava_2004,Rowstron_2001,Malkhi_2002,Naor_2007,Stoica_2001}). A standard approach to design a distributed data structure that is provably robust to a large number of faults is to define an underlying network with good structural properties (e.g., expansion, low diameter, etc.) and efficient distributed algorithms able to quickly restore the network and data structure after a certain amount of nodes or edges have been adversarially (or randomly) removed (e.g., see \cite{Pandurangan_2001,Kuhn_2010,Jacobs_2013,Becchetti_2023}). Most prior works develop algorithms that will work under the assumption that the network will eventually stabilize and stop changing or that an overall (somehow) limited amount of faults can occur. 

Distributed Hash Tables (DHTs) (see for example \cite{Rowstron_2001_a,Rowstron_2001_b,Zhao_2002,Stoica_2003,Kaashoe_2003,Manku_2003,Ganesan_2004}) are perhaps the most common distributed data structures used in P2P networks. A DHT scheme~\cite{Lua_2005} creates a fully decentralized index that maps data items to peers and allows a peer to search for an item efficiently without performing flooding. Although DHT schemes have excellent congestion properties, these structures do not allow for non-trivial queries on ordered data such as nearest-neighbor searching, string prefix searching, or range queries. 

To this end, Pugh~\cite{Pugh_1990} in the $90$'s introduced the skip list, a randomized balanced tree data structure that allows for quickly searching ordered data in a network. Skip lists have been extensively studied~\cite{Papadakis_2090,Devroye_1992,Kirschenhofer_1994,Kirschenhofer_1995} and used to speed up computation in centralized, (batch) parallel and distributed settings~\cite{Gabarro_1994,Gabarro_1997,Pugh_1998,ShavitL_2000,Sundell_2004,Fraser_2004,Herlihy_2006,Tseng_2019}. However, classical skip lists especially when implemented on a distributed system do not deal with the chance of having failures due to peers (elements) abruptly leaving the network (a common feature in P2P networks). 

With the intent of overcoming such a problem, Aspnes and Shah~\cite{Aspnes_2007} presented a distributed data structure, called a \emph{skip graph} for searching ordered data in a P2P network, based on the \emph{skip list} data structure~\cite{Pugh_1990}. Surprisingly, in the same year, Harvey et al.~\cite{Harvey_2003} independently presented a similar data structure, which they called SkipNet. Subsequently, Aspnes and Wieder~\cite{Aspnes_2009} showed that skip graphs have $\Omega(1)$ expansion with high probability (\whp). 
Although skip graphs enjoy such resilience property, the only way to fix the distributed data structure after some faults is either (i) use a repair mechanism that works only in the absence of new failures in the network and has a linear worst-case running time\footnote{In the size of the skip graph.}~\cite{Aspnes_2007} or (ii) rebuild the skip graph from scratch. Goodrich et al.~\cite{Goodrich_2006} proposed the \emph{rainbow skip graph}, an augmented skip graph that enjoys lower congestion than the skip graph. Moreover, the data structure came with a periodic failure recovery mechanism that can restore the distributed data structure even if each node fails independently with constant probability. More precisely, if $k$ nodes have randomly failed, their repair mechanism uses $\Oo(\min(n,k\log n))$ messages over $\Oo(\log^2 n)$ rounds of message passing to adjust the distributed data structure. In the spirit of dealing with an efficient repairing mechanism, Jacob el at.~\cite{Jacob_2014} introduced \skippl, a self-stabilizing protocol\footnote{We refer to~\cite{Dolev_2004} for an in-depth description of self-stabilizing algorithms.} that converges to an augmented skip graph structure\footnote{It is augmented in the sense that it can be checked \emph{locally} for the correct structure.} in $\Oo(\log^2 n)$ rounds \whp, for any given initial graph configuration in which the nodes are weakly connected. The protocol works under the assumption that starting from the initial graph until the convergence to the target topology, no external topological changes happen to the network. Moreover, once the desired configuration is reached, \skippl can handle a \emph{single} join or leave event (i.e., a new node connects to an arbitrary node in the system or a node leaves without prior notice) with a polylogarithmic number of rounds and messages. While it is shown that these data structures can tolerate node failures, there is no clarity on how to handle persistent churn wherein nodes can continuously join and leave, which is an inherent feature of P2P networks. Moreover, all the proposed repairing mechanisms ~\cite{Aspnes_2007,Aspnes_2009,Goodrich_2006,Jacob_2014}  will not work in a highly dynamic setting with \emph{large, continuous, adversarial} churn (controlled by a powerful adversary that has full control of the network topology, including full knowledge and control of what nodes join and leave and at what time and has unlimited computational power). 

%% file: trunk/conclusion.tex
\section{Concluding remarks and discussion}
In this work we proposed the first churn resilient skip list that can tolerate a heavy adversarial churn rate of $\Oo(n/\log n)$ nodes per round. The data structure can be seen as a four networks architecture in which each network plays a specific role in making the skip list resilient to churns and keeping it continuously updated. Moreover, we provided efficient $\Oo(\log n)$ rounds resource competitive algorithms to (i) delete a batch of elements from a skip list (ii) create a new skip list and, (iii) merge together two skip lists. This last result is the first algorithm that can merge two skip lists (as well as a skip list and a batch of new nodes) in $\Oo(\log n)$ rounds \whp. We point out that these algorithms can be easily adapted to work on skip graphs~\cite{Aspnes_2007,Goodrich_2006,Jacob_2014}. 

In a broader sense, our technique is general and can be seen as a framework to 
maintain any kind of distributed data structure despite heavy churn rate. The only requirement is to devise efficient \emph{delete, buffer creation, merge}, and \emph{update} algorithms for the designated data structure. Furthermore, this allows us to define complexity classes for the maintenance of distributed data structure in the Dynamic Networks with Churn Model. Indeed, in this work we showed that skip list and skip graphs belong to the class of data structure that can tolerate a churn rate of $\Oo(n/\log n)$ per round.

An additional contribution of our work is the improvement on the $\Oo(\log^3 n)$ rounds state-of-the-art technique for sorting in the \nccz~model. We show how to sort $n$ elements (despite a high churn rate) in $\Oo(\log n)$ rounds using results from sorting network theory. However, given the impracticability of the AKS sorting network, our result is purely theoretical. In practice, it could be easily implemented using Batcher's network instead of the AKS one. This change would slow down the bootstrap phase and the maintenance cycle to $\Oo(\log^2 n)$ rounds. The churn rate that we can tolerate will only drop down to $O(n/\log^2 n)$.

Finally, given the simplicity of our approach, we believe that our algorithms could be used as building blocks for other non-trivial distributed computations in dynamic networks.

%% file: trunk/appendix.tex
\section{Useful Mathematical Tools}\label{sec:appendix}
In Appendix~\ref{apx:reshaping_protocol}, we use the Chernoff bound for the Poisson trials
\begin{theorem}[Theorem 4.4~\cite{Mitzenmacher_2017}]\label{thm:chernoff_posson_trials}
Let $X$ be a sum of $n$ independent Poisson trials $X_i$ such that $\prob{X_i = 1}=p_i$, for $i\in [n]$. Then,
\begin{align}
   & \prob{X\geq R}\leq 2^{-R}&\text{for }R\geq 6\expect{X}&
\end{align}
\end{theorem}
Moreover, to provide high confidence bounds for the properties of randomized skip list we use the following tools.
\begin{theorem}[Right tail]
	Let $X$ be a sum of $n$ independent random variables $X_i$ such that $X_i\in [0,1]$. Let $\expect{X} = \mu$. Then,
	\begin{align}
		&\prob{X\geq k}\leq \left(\frac{\mu}{k}\right)^k\left(\frac{n-\mu}{n-k}\right)^{n-k}\leq \left(\frac{\mu}{k}\right)^k e^{k-\mu} &\text{for }k>\mu &\\
		&\prob{X\geq (1+\varepsilon) \mu}\leq \left(\frac{e^\varepsilon}{(1+\varepsilon)^{1+\varepsilon}} \right)^\mu  &\text{for }\varepsilon \geq 0&
	\end{align}
\end{theorem}

\begin{theorem}[Left tail]
	Let $X$ be a sum of $n$ independent random variables $X_i$ such that $X_i\in [0,1]$. Let $\expect{X} = \mu$. Then,
	\begin{align}
		&\prob{X\leq k}\leq \left(\frac{\mu}{k}\right)^k\left(\frac{n-\mu}{n-k}\right)^{n-k}\leq \left(\frac{\mu}{k}\right)^k e^{k-\mu} &\text{for }k<\mu &\\
		&\prob{X\leq (1-\varepsilon) \mu}\leq e^{-\frac{\mu\varepsilon^2}{2}}  &\text{for }\varepsilon \in (0,1)&
	\end{align}
\end{theorem}
Furthermore, a useful result about the upper tail value of a negative binomial distribution to a lower tail value of a suitably defined binomial distribution that allows us to use all the results for lower tail estimates of the binomial distribution to derive upper tails estimate for negative binomial distribution. This a very nice result because finding bounds for the right tail of a negative binomial distribution directly from its definition is very difficult.
\begin{theorem}[See Chapter $4$ in~\cite{Ross_1976}]\label{thm:neg_binom}
	Let $X$ be a negative binomial random variable with parameters $r$ and $p$. Then, $\prob{X>n} = \prob{Y<r}$ where $Y$ is a binomial random variable with parameters $n$ and $p$.
\end{theorem}
\section{Useful Skip list properties}\label{apx:skip_list}
\begin{lemma}\label{lemma:exp_hight}
	The height of a $n$-element skip list is $\Oo(\log n)$ \whp
\end{lemma}
\begin{proof}
	Let $Y_i$ for $i\in [n]$ be the random variable that counts the number of consecutive heads we obtain while tossing a p-biased coin before we get a tail. Moreover, define $h$ to be the maximum height of the skip list i.e., $h = 1+\max\{Y_i: i\in [n]\}$. Observe that $\prob{Y_i\geq k} = p^{k-1}$, and that by a straightforward application of the union bound we obtain the probability of having a skip list of height at least $k$, $\prob{h\geq k}\leq np^{k-1}$. Choosing $k = 4\log_{1/p}n+1$, we obtain a high confidence bound for the number of levels in the distributed data structure 
    \begin{align}
    \prob{h\geq 4\log_{1/p}n+1}\leq n(1/n)^{4} = 1/n^3
    \end{align}
	Moreover, $h$'s expected values is 
	\begin{align}
		&\expect{h} = \sum_{i>0}\prob{h\geq i} = \sum_{i= 1}^{4\log n}\prob{h\geq i} + \sum_{i>4\log n}\prob{h\geq i}\leq \sum_{i = 1}^{4\log n}1 + \sum_{i>4\log n}np^{i-1}&\\
		&\leq 4\log n +np^{4\log n}\left(\sum_{i\geq 1}p^{i-1}\right) = 4\log n+n \left(\frac{1}{n^4}\right)\left(\frac{1}{1-p}\right)\leq 4\log n +1 = \Oo(\log n)&
	\end{align}
	Where the last inequality holds if the $\log$ is in base $1/p$ and if $n$ is sufficiently large.
\end{proof}
Next, we say that $X_\ell=\{x_1,x_2,\dots ,x_k\}$ is a \emph{run}/cohesive group of nodes in a skip list $\Ll$ at level $\ell$, if $X_\ell$ is a set of consecutive nodes such that $\lmax{\Ll}{x_i}=\ell$ for each $x_i\in X_\ell$. Moreover, we give a high confidence bound on the size of a run $X_\ell$.
\begin{lemma}\label{lemma:run_length}
    The size of a run of nodes $X_\ell$ for some level $\ell$, is at most $\Oo(\log n)$ \whp, and its expected value is $1/(1-p)$.
\end{lemma}
\begin{proof}
The size of a run of nodes $|X_\ell|$ is a geometric random variable with parameter $1-p$. Thus, 
    \begin{align}
        \prob{|X_\ell| \geq k} \leq (1-p)^{k-1}
    \end{align}
Choosing $k = c \log_{1-p} n+1$ for $c\geq 1$, gives us a high confidence bound on the number of consecutive nodes in a run of node $X_\ell$ at some level $\ell$. Indeed,
    \begin{align}
        \prob{|X_\ell| \geq c \log_{1-p} n+1}\leq 1/n^c
    \end{align}
     Moreover, its expected value is $\expect{|X_\ell|} = 1/(1-p)$, assuming $p=1/2$ (in other words, constant), we have that the expected length of a run of nodes at some level $\ell$ of the skip list is $\Oo(1)$.
\end{proof}
Finally, we conclude with the analysis of the search/deletion/insertion of an element in a skip list (see Figure~\ref{fig:searc_example} for an example about one of these operations).
\begin{lemma}\label{lemma:classic_insertion}
    The running time for an insertion, deletion and search of an element in a skip list takes $\Oo(\log n)$ rounds \whp
\end{lemma}
\begin{proof}

	Let $R_{i}$ be number of horizontal edges at level $0\leq i \leq h$ crossed by a search operation that starts on the left topmost node of the skip list. Define the random variable $W_h = |R_{0}| +|R_{1}|+\dots+|R_{h}|$ to be the amount of horizontal moves performed by the search algorithm. Observe that each $R_{i}$ is a geometric random variable of parameter $1-p$ and that $h$ is a random variable itself. From Lemma~\ref{lemma:exp_hight} we know that $\expect{h} = \Oo(\log n)$ and that $\prob{h\geq 4\log n +1}\leq \frac{1}{n^3}$. 
	Since $h$ is a random variable, $W_h$ is a \emph{random sum} of random variables thus we can not make a straightforward use of the properties of sum of geometric random variables, rather we can write:
	\begin{align}
		&\prob{W_h> 16\log n} = \prob{W_h>16\log n \cap h\leq 4\log n} + \prob{W_h>16\log n \cap h> 4\log n} & \\ &\leq\sum_{h = 0}^{4\log n}\prob{W_h>16\log n}+\prob{h> 4\log n} \leq (1+4\log n)\prob{W_{4\log n}>16\log n}+\frac{1}{n^3}&
	\end{align}
	We notice that $W_{4\log n}$ is a deterministic sum of geometric random variables. Hence we can use the relation between the upper tail value of a negative binomial distribution and the lower tail value of a suitably defined binomial distribution to derive upper tail estimates for the negative binomial distribution (see Appendix~\ref{sec:appendix}). Thus, we rewrite $\prob{W_{4\log n}>16\log n}$ as $\prob{Y\leq 4\log n}$ where $Y$ is a binomial random variable with parameter $n=16\log n$ and $p$. Now, we apply a Chernoff bound~\cite{Mitzenmacher_2017} and by setting $k=4\log n$, $p = 1/2$ and $\mu = 8\log n$ we obtain 
	\begin{align}
		&\prob{W_{4\log n}>16\log n} = \prob{Y\leq 4\log n}\leq \left(\frac{8\log n}{4\log n}\right)^{4\log n}\left(\frac{8\log n}{12\log n}\right)^{12\log n} = \frac{2^{16\log n}}{3^{12 \log n}} &\\ &=\left(\frac{16}{27}\right)^{4\log n} \leq \frac{1}{n^3}&
	\end{align}
	Therefore, 
	\begin{align}
			\prob{W_{h}>16\log n} \leq (1+4\log n)\left(\frac{1}{n^3}\right)+\frac{1}{n^3}< \frac{1}{n^2}\quad \text{ if }n>32
	\end{align}
	Now that we have derived a high confidence bound for $W_h$ we can obtain a bound for its expected value:
	\begin{align}
	\expect{W_h} = \sum_{i= 1}^{16\log n}\prob{W_h\geq i}+\sum_{i>16\log n}\prob{W_h\geq i}\leq 16\log n + c = \Oo(\log n)
	\end{align}
	Observe that the first sum is bounded above $16\log n$ because every probability is less than $1$ and the second one is dominated by $\sum_{i\geq 1}1/i^2$ which is a constant.
	Since the running time of a search operation is bounded by the number of horizontal moves performed at each level plus the number of vertical moves to reach the target node from the topmost level to the bottom most. We have that the overall running time of a search procedure starting at the left topmost node in the skip list is $T=h+W_h = \Oo(\log n)$ with probability at least $1-(1/n^2)$, by applying the union bound we have that $T$ is $\Oo(\log n)$ with probability at least $1-(1/n)$ starting at any node in the skip list. To conclude, the expected running time is $\expect{T} =\expect{h}+ \expect{W_h}=\Oo(\log n)$.
\end{proof}

\begin{figure}[htb!]
	\centering
    \includegraphics[scale=0.5]{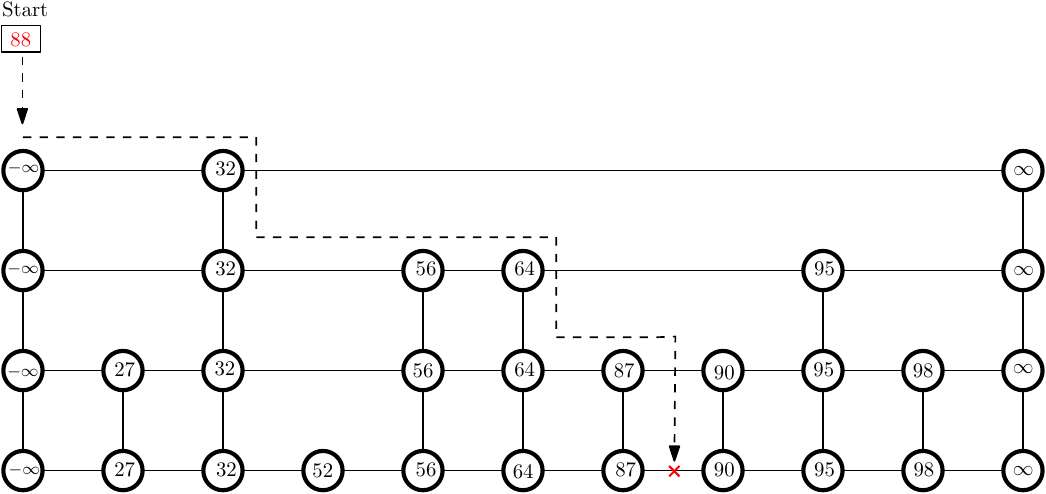}
	\caption{Example of the insertion of the element $88$ in the skip list. Dashed line is the search path. }\label{fig:searc_example}
\end{figure}
\section{Spartan's Reshaping protocol}\label{apx:reshaping_protocol}

\begin{figure}[htb!]
	\centering
	\includegraphics[scale=0.7]{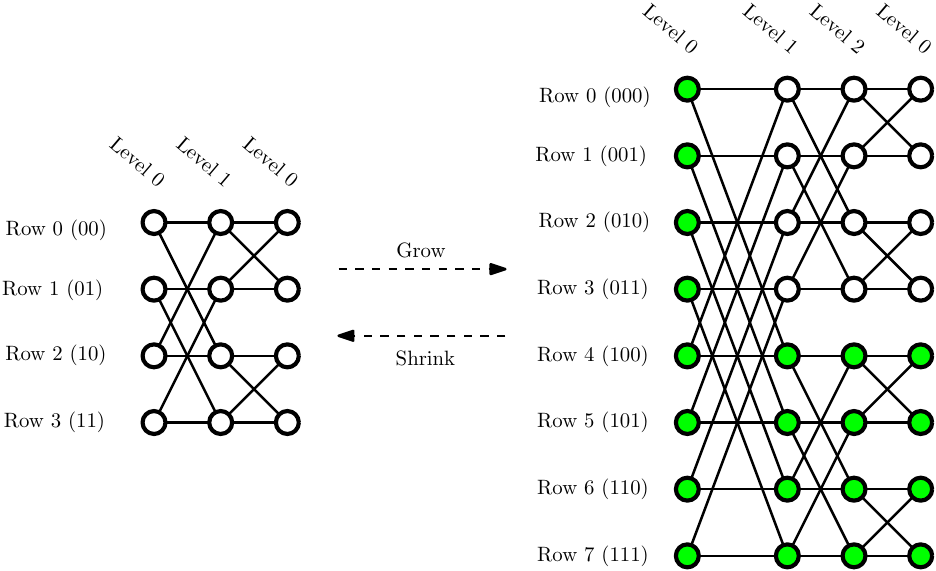}
	\caption{Example of the reshaping procedure. On the left, the \Spartan~structure as a one-dimensional wrapped butterfly of committees. The rows are represented with their binary representations. Every node in the figure is a random committee of $\Theta(\log n)$ nodes and every edge between two committees encodes a complete bipartite graph. On the right, $2$-dimensional butterfly is obtained by increasing the dimensionality of the left-wrapped butterfly by one. Green nodes are the ones that must be added to increase the dimensionality of the $1$-dimensional wrapped butterfly. Moreover, left-to-right execution increases the dimensionality by one, and right-to-left execution decreases the dimensionality by one.}\label{fig:grow}
\end{figure}

In this section, we provide a detailed description of the reshaping protocol mentioned in Section~\ref{sec:churn_resilient_network}.
Let $\alpha,\beta \in \mathbb{R}_{\geq 0}$ such that $\alpha>1$ and $\beta\in (0,1)$ be two constants. We want the Spartan network $\X$ to be enlarged if all the committees have size at most $\alpha c\log n$, to shrink if all the committees have size at most $\beta c\log n$, or stay the same otherwise ($c>0$ is the constant hidden by the asymptotic notation). In other words, we want $\X$ to be enlarged or reduced when the wrapped butterfly network size increases or decreases by one dimension respectively. Moreover, a general technique to increase the dimensionality of a $k$-dimensional butterfly network $\X$ by one is to: (1) create a copy of $\X$ in which each row ID $(i,j)$ is ``shifted'' by $2^{k}$ i.e., $(2^{k}+i,j)$; (2) ``shift'' the index of the levels of the two disjoint butterflies by one; (3) add a new row of $2^{k+1}$ nodes to form the level $0$ of the $k+1$ dimensional butterfly; and, (5) connect the newly created level $0$ to the two disjoint butterflies using the ``standard'' butterfly edge-construction rule~\cite{Leighton_2014}. Furthermore, to reduce the dimensionality of a $k$-dimensional butterfly, we need to remove the level $0$, remove one of the two disjoint $k-1$-dimensional butterflies, and fix the row and levels IDs if needed. Figure~\ref{fig:grow} shows an example on how to perform such operations. 
We start with the description of the growing procedure. A new number of nodes $n' = \alpha n$ is reached such that $2^k k \in \Oo(N)$ must be increased to $2^{k+1} (k+1) \in \Oo(N')$ (the number of committees increases from $N=\Oo(n/\log n)$ to $N' =\Oo(n'/\log n')$)
and update the butterfly structure according to the new size of the network. Moreover, each step of Algorithm~\ref{algo:reshape} requires at most $\Oo(\log n)$ rounds and no node will send or receive more than $\Oo(\log n)$ messages at any round. 
As a first step, the committee $C(0,0)$ is elected as a leader\footnote{We can make this assumption because we do not deal with byzantine nodes.}, and each committee routes a message in which it expresses its opinion (grow, shrink, stay the same) to the leader committee (lines 1-2). This can be done very efficiently on hypercubic networks~\cite{Leighton_2014,Mitzenmacher_2017}, indeed we have that the leader committee will receive all the opinions in $\Oo(\log N)$ rounds~\cite{Leighton_1994} where $N\in\Oo(n/\log n)$ is the number of committees in the wrapped butterfly. If the network agrees on growing, the protocol proceeds with increasing the dimensionality of the wrapped butterfly by one. Moreover, assume that the $k$-dimensional butterfly must be transformed into a $k+1$-dimensional one. To this end, the protocol executes the growing approach described before. Every committee shifts its level ID by one and promotes two random nodes in each committee to be the committee leaders of the copy of the butterfly and of the new level $0$. Next, all the newly elected committee leaders will start recruiting random nodes in the network until they reach a committee size of $\Theta(\log n')$ where $n' = \alpha n$. Finally, each new committee will create edges according to the wrapped butterfly construction algorithm~\cite{Leighton_2014,Augustine_2021} and drop the old ones (if any). The next lemma shows that the growing phase requires $\Oo(\log n')$ rounds \whp~where $n'=\alpha n$.
\begin{lemma}\label{lemma:spartan_growing}
	The network growth process ensures that every committee has $\Theta(\log n')$ members after $\Oo(\log n')$ rounds \whp
\end{lemma}
\begin{proof}
	Let $b = \frac{3n'}{4N'} =\frac{3}{4}\log n'$ we show that after the first phase of the process, each new committee with gain $b$ nodes. 
	During each round, there are at least $\frac{n'}{4}-N'\in\Omega(n')$ nodes that did not receive any invitation from one of these new committee leaders. Moreover, define the event $E_v=\text{``The committee leader \textit{v} recruits a node''}$ then setting $N'=\frac{n'}{c\log n'}$ with $c\geq 1$ gives $\prob{E_v}\geq \left(\frac{n'/4-N'}{n'}\right)\left(1-\frac{1}{n'}\right)^{N'}\geq \frac{3}{16}$. To bound the expected time needed by a committee to recruit $b$ elements we define a pure-birth Markov Chain $\{X_t\}_{t\geq 0}$ with state space $\Omega = \{0,1,2,\dots ,b\}$ and initial state $0$ that counts the number of recruited members. At each round, the Markov Chain at state $i<b$ can proceed one step forward to state $i+1$ with probability $p_{i,i+1}=\prob{E_v}$ or loop on $i$ with probability $r_{i,i} = 1-p_{i,i+1}$. Observe that each state $0\leq i<b$ is a transient state, while $b$ is an absorbing state, and that the expected absorption time in state $b$ is $\Oo(b)=\Oo(\log n')$. To give a tail bound on the absorption time, we study a more pessimistic random process in which we toss a p-biased coin with $p = 3/16$ and count the number of rounds $Z_i$ before we get $b$ heads. In this experiment, the expected number of rounds needed to get $b$ heads is $\frac{16}{3}b\in \Theta(\log n')$. By applying a Chernoff bound (Theorem~\ref{thm:chernoff_posson_trials}) we can show that for any value of $n'$ and any fixed $R$, such that $R\geq 6\expect{Z}$ the probability that we will need more than $R\log n'$ rounds to get $b$ number of heads is at most $1/n'^R$. In other words, this means that the Markov Chain will take at most $\Oo(\log n')$ rounds with probability $1-1/n'^R$ to reach state $b$. Finally, by applying the union bound we can show that every committee will need $\Oo(\log n')$ rounds with probability $1-1/n'^{R-1}$ to recruit $b$ members.
\end{proof}
Furthermore, in the case in which the network agreed on shrinking i.e., reducing the butterfly dimensionality from $k$ to $k-1$ a pool of vacated nodes $\V$ is created and all the nodes in the level $0$ and in the committees $C(i,j)$ such that $2^{k-1} \leq i<2^k$ and $1\leq j\leq  k$ join such pool of nodes. Next, each committee in the remaining butterfly network will ``recruit'' $b = \frac{3}{4} \log n'$ (here $n' = \beta n$) random nodes from $\V$. After the recruitment phase, each remaining node $u$ in $\V$ (if any) will probe a random committee in the butterfly. If such a committee has not reached the target size yet it accepts the request, otherwise $u$ tries again with a different committee. It follows that the shrinking phase can be done in $\Oo(\log n')$ rounds \whp~where $n'=\beta n$.
\begin{lemma}\label{lemma:spartan_shrinking}
	The network shrinking process ensures that every committee has $\Theta(\log n')$ members after $\Oo(\log n')$ rounds \whp
\end{lemma}
\begin{proof}
	Let $b = \frac{3n'}{4N'} =\frac{3}{4}\log n'$ and observe that the analysis of the recruitment phase is identical to the one for the growing process. Moreover, after $\Theta(\log n)$ rounds, there can be $n'/4$ nodes in $\V$ that are not part of any committee. Each node that joined a committee has one unit of budget that can use to induce a new member into its committee. Every remaining node $u\in \V$ probes a random node $v$ asking to join its committee, if $v$'s budget is greater than $0$ then $v$ accepts the request, otherwise $u$ tries again with a different node. Furthermore, observe that even if all the remaining nodes $\V$ found a spot in some committee, there will be $n'/2$ nodes that would not have exhausted their budget and each remaining node can find a committee with probability at least $1/2$. Finally, by using similar arguments to the growing process we can show that each remaining node $u$ succeed in finding a committee with high probability in $\Theta(\log n')$ rounds. And by using the union bound we can guarantee that each remaining node can find and become part of a committee in $\Theta(\log n')$ \whp
\end{proof}
\input{pseudocodes/grow}

\section{Detailed execution of the WAVE Protocol}\label{apx:wave}
In this section we show a detailed execution of the WAVE protocol. The figures below describe the behavior of the nodes in the clean network $\C$ and the behavior of the nodes in the buffer network $\B$ during the execution of the protocol. We omit the buffer network description when the states of the nodes in $\B$ are clear from the context.

\begin{figure}[htb!]
	\captionsetup[subfigure]{justification=centering}
	\centering
	\begin{subfigure}{0.45\textwidth}
		\centering
		\includegraphics[scale=0.7]{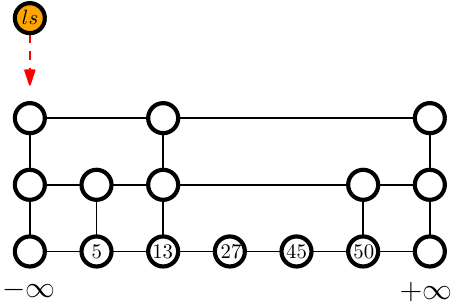}
		\caption{}\label{fig:wave_1_r1}
	\end{subfigure}
	\begin{subfigure}{0.44\textwidth}
		\centering
		\includegraphics[scale=0.7]{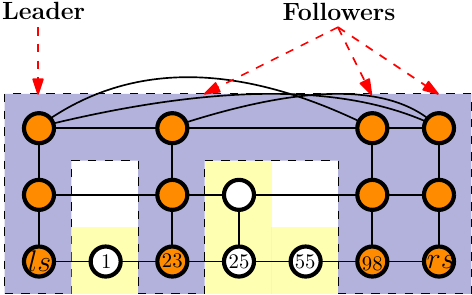}
		\caption{}\label{fig:wave_2_r1}
	\end{subfigure}
    \caption{Figure~\ref{fig:wave_1_r1} shows the Clean Network $\C$ while Figure~\ref{fig:wave_2_r1} shows the Buffer Network $\B$ after the preproccessing phase. At the beginning, in $\B$, the only cohesive group in the merge state is composed by the nodes in $\B$'s top-most level, with $ls$ (that is the $\B$'s left sentinel) as a leader. $ls$ is the node that traverses $\C$ while representing the whole cohesive group $X = \{ls,23,98,rs\}$.}
\end{figure}

\begin{figure}[htb!]
	\captionsetup[subfigure]{justification=centering}
    \centering
	\begin{subfigure}{0.47\textwidth}
		\centering
		\includegraphics[scale=0.7]{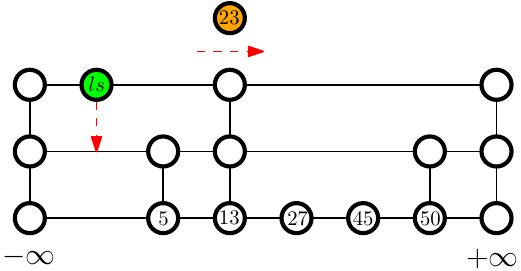}
		\caption{}\label{fig:wave_3_r2}
	\end{subfigure}
	\begin{subfigure}{0.44\textwidth}
		\centering
		\includegraphics[scale=0.7]{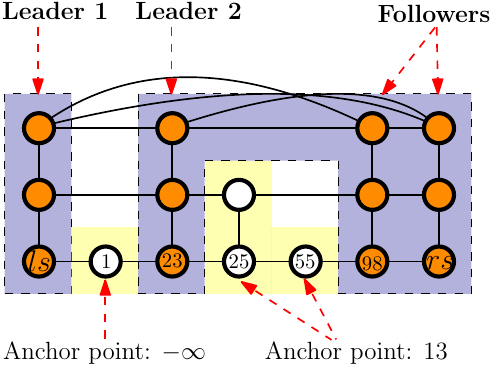}
		\caption{}\label{fig:wave_2_r2}
	\end{subfigure}
    \caption{The \emph{unique} cohesive group $X$ splits into two cohesive groups: $X_{\texttt{down}} = \{ls\}$ and $X_{\texttt{right}} = \{23,98,rs\}$. $X_{\texttt{down}}$ merges in level $3$ and proceeds downwards while $X_{\texttt{right}} $ (lead by the new $23$ as the new leader) moves right to find its designated spot to merge at level $3$. Every step executed by $X_{\texttt{down}}$ and $X_{\texttt{right}}$ on $\C$ is communicated to their respective children that, upon receiving new information, update their anchor points (i.e., their virtual walk position).  }
\end{figure}

\begin{figure}[htb!]
	\captionsetup[subfigure]{justification=centering}
    \centering
	\begin{subfigure}{0.6\textwidth}
		\includegraphics[scale=0.7]{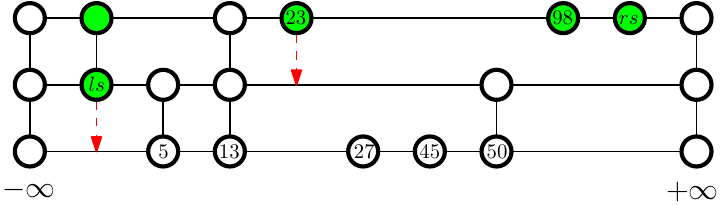}
		\caption{}\label{fig:wave_3_r3}
	\end{subfigure}
	\begin{subfigure}{0.35\textwidth}
		\includegraphics[scale=0.7]{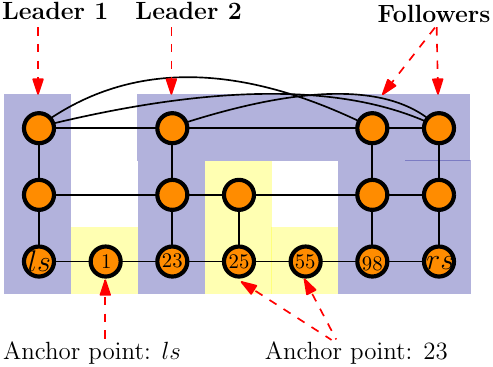}
		\caption{}\label{fig:wave_2_r3}
	\end{subfigure}
        \caption{ $X_{\texttt{down}} = \{ls\}$ keeps merging at level $2$ while $X_{\texttt{right}} = \{23,98,rs\}$ merges at level $3$ between $13$ and $+\infty$ and moves one step down. Their children update their virtual walk status as well.  }
    \end{figure}

\begin{figure}[htb!]
	\captionsetup[subfigure]{justification=centering}
    \centering
	\begin{subfigure}{0.9\textwidth}
		\centering
		\includegraphics[scale=0.8]{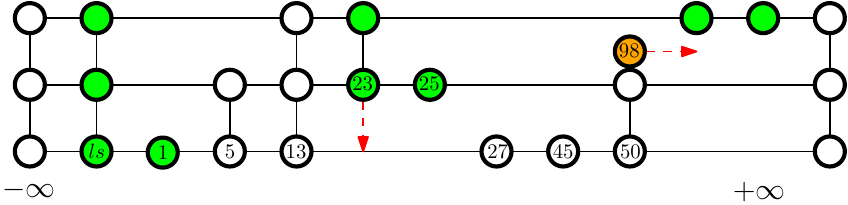}
		\caption{}\label{fig:wave_3_r4}
	\end{subfigure}
    
	\begin{subfigure}{0.44\textwidth}
		\centering
		\includegraphics[scale=0.8]{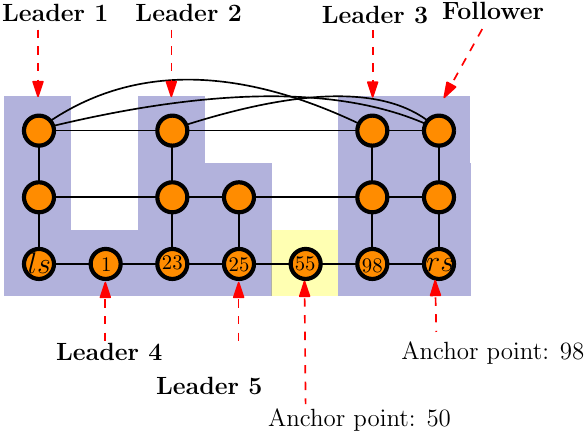}
		\caption{}\label{fig:wave_2_r4}
	\end{subfigure}
    \caption{$X_{\texttt{down}} = \{ls\}$ completed the merging in $\C$, consequently the node $1$ starts its own merge phase at $ls$ in level $0$ and in $\Oo(1)$ rounds finds its spot in the base level of $\C$ completing its own merge phase. $X_{\texttt{right}} = \{23,98,rs\}$ splits in two,  $X_{\texttt{down}}' = \{23\}$ and $X_{\texttt{right}}' = \{98,rs\}$ where $X_{\texttt{down}}'$ continues merging downwards while $X_{\texttt{right}}'$ moves right in level $1$. Finally, node $25$ starts its merge phase at level $1$ starting from $23$ in $\C$, and, being dependent on $23$, waits for $25$ to merge in level $0$ before proceeding. }
       \end{figure}
\begin{figure}[htb!]

\centering
\includegraphics[scale=0.8]{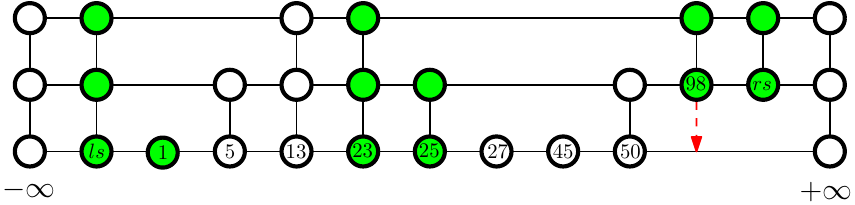}
\caption{$X_{\texttt{down}}' = \{23\}$ complete its own merge phase by merging at level level $0$. Consequently, $25$ merges in level $0$ in $\Oo(1)$ rounds. On the other side of the skip list, $X_{\texttt{right}}' = \{98,rs\}$ merges at level $1$ and jumps one level below, while $5$ ``virtually'' waits for its turn at node $50$ in $\C$.}\label{fig:wave_3_r5}
\end{figure}

\begin{figure}[htb!]
\centering
\includegraphics[scale=0.8]{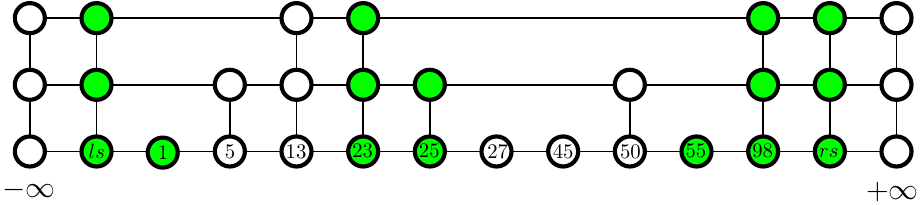}
\caption{$X_{\texttt{right}}' = \{98,rs\}$ completes its merge phase by merging in the bottom most level,  $5$ enters the merge phase starting at $50$ and successfully merges in $\Oo(1)$ rounds in level $0$.}\label{fig:wave_3_r6}
\end{figure}

%% file: pseudocodes/grow.tex
\begin{algorithm}[htb!]
	\caption{Reshape protocol. }
	\label{algo:reshape} 
	\nonl \textbf{Agreement Phase.}
	
	The committee $C(0,0)$ is elected as a leader $\ell$.
	
	Every committee in $\X$ routes its opinion to $\ell$.
	
	\uIf{Agreement on Growing}{
	
	\nonl \textbf{Growing Phase.}\tcp*{Here $n' = \alpha n$}
	
	Each committee $C(i,j)$ for $0\leq i < 2^k$ and $0\leq j <  k$:

		\Indp{\nonl \textbf{(a)} Increases its committee ID from $C(i,j)$ to $C(i,j+1)$.}\\
  
		{\nonl \textbf{(b)} Promotes a random node $u$ among the available committee members to be the leader of the new committee $C(j\cdot 2^{k-1}+i,0)$.}\\
  
		  {\nonl \textbf{(c)} Promotes a random node $v$ among the available committee members to be the leader of the new committee $C(2^{k-1}+i,j+1)$.}\\

	\Indm Each newly generated committee leader actively recruits new nodes in their committee until it gets $\Theta(\log n')$ committee members.
	
	Each new committee creates complete bipartite edges according to the Wrapped Butterfly construction algorithm (see~\cite{Leighton_2014,Augustine_2021}).
}\uElseIf{Agreement on Shrinking}{
	\nonl \textbf{Shrinking Phase.} \tcp*{Here $n' = \beta n$}
	
	Each member in the committees $C(i,0)$ for $0\leq i < 2^k$ and $C(i,j)$ for $2^{k-1}\leq i < 2^k$ $1\leq j \leq k$ ``vacate'' its committee and joins a pool of unassigned nodes.	
	
	Each node in the pool of unassigned nodes will randomly join one of the $N' =\Oo(n'/\log n')$ new committees such that each committee has $\Theta(\log n')$ nodes.
}\Else{

	Do nothing.
}
\end{algorithm}